\documentclass[journal, final]{IEEEtran}
\usepackage{amsfonts}
\usepackage{amsmath, amssymb}
\usepackage{epsfig}
\usepackage{color}
\newtheorem{thm}{Theorem}[section]

\newtheorem{lem}{Lemma}[section]

\title{On the Capacity of a Class of MIMO\\ Cognitive Radios}
\author{Sriram Sridharan,~\IEEEmembership{Student Member,~IEEE} and Sriram Vishwanath,~\IEEEmembership{Member,~IEEE}%
\thanks{Manuscript received May 18, 2007; revised September 16, 2007; revised October 26, 2007. This research was supported in part by National Science Foundation
grants NSF CCF-0448181, NSF CCF-0552741, NSF CNS-0615061, and NSF
CNS-0626903, THECB ARP and the Army Research Office YIP. The material in this paper was presented in part at the IEEE Information Theory Workshop, Lake Tahoe, CA, September 2007 \cite{SridharanVishwanath2007}.}%
\thanks{The authors are with the Wireless Networking and Communications Group, Department of Electrical and Computer Engineering, University of Texas at Austin, Austin, TX - 78712 (email: sridhara@ece.utexas.edu; sriram@ece.utexas.edu).}%
}
\begin{document}
\maketitle
\begin{abstract}
Cognitive radios have been studied recently as a means to utilize
spectrum in a more efficient manner. This paper focuses on the
fundamental limits of operation of a MIMO cognitive radio network
with a single licensed user and a single cognitive user. The channel
setting is equivalent to an interference channel with degraded
message sets (with the cognitive user having access to the licensed
user's message). An achievable region and an outer bound is derived
for such a network setting. It is shown that under certain conditions,
the achievable region is optimal for a portion of the capacity region that includes sum
capacity.
\end{abstract}
\section{Introduction}
The design of radios to be ``cognitive"  has been identified by the
Federal Communications Commission (FCC) as the next big step in
better radio resource utilization \cite{FCC}. The term ``cognitive"
has many different connotations both in analysis and in practice,
but with two underlying common themes: {\it intelligence} built into
the radio architecture coupled with {\it adaptivity}.

Cognitive radios have been studied under different model settings. The first models studied cognitive radios as a spectrum sensing problem \cite{Mitola2000}\cite{Ghasemi2005}\cite{Haykin2005}\cite{Srinivasa2006}. Under this setting, the cognitive radio opportunistically uses licensed spectrum when the licensed users are sensed to be absent in that band. Problems encountered in this setup are threefold :
\begin{enumerate}
\item Sensing must be highly accurate to guarantee non interference with the licensed radio.
\item Control and coordination between the cognitive transmitter receiver pair is required to ensure the same spectrum is used, and finally
\item There are no QoS guarantees for the cognitive transmitter receiver pair.
\end{enumerate}
Other models with different side information at the cognitive users have been studied. In \cite{Jafar2007} and \cite{Jafar2007a}, the authors study frequency coding by the cognitive transmitter by assuming non causal knowledge of the frequency use of the primary transmitter.

In this paper, we study cognition from an information theoretic setting where we assume that the cognitive transmitter knows the message of the licensed transmitter apriori. Such a model is interesting for two reasons : 1)  It provides an
upper limit, or equivalently a benchmark on the performance of systems where the cognitive radio gains a partial understanding of the licensed transmitter and 2) It allows us to understand the ultimate limits on the cognitive transmitter by giving it maximum information and allowing it to change its transmission and coding strategy based on all the information available at the licensed user. In essence, it enlarges the possible schemes that can be implemented at the cognitive radio, and 3) It lends itself to information theoretic analysis, being a setting where such tools can be applied to determine the performance limits of the system. Many other configurations, including the interference channel setting when the cognitive transmitter does not know the message of the licensed transmitter are multi-decade long open problems.

The goal of this paper is to study the fundamental limits of performance of cognitive radios. Along the lines of
\cite{Devroye2006a}, we consider the model depicted in Figure
$1$. In this setting, we have an interference
channel \cite{Carleial1978}\cite{Costa1985}\cite{Sato1981}\cite{Han1981},
but with degraded message sets, where the transmitter with a single
message is called ``legacy," ``primary" or ``dumb" and the
transmitter with both messages termed the ``cognitive" transmitter. Prior
work on this model for the single antenna case is in
\cite{Devroye2006a}\cite{Maric2005a}\cite{Jovicic2006}\cite{Wu2006c}.

In this paper, we study the performance of the cognitive radio model under a multiple antenna (MIMO) setting. Both the licensed and cognitive transmitter and receiver may have multiple antennas. MIMO is fast becoming the most common feature of wireless systems due to its performance benefits. Thus, it is important to study the capacity of cognitive radios under a MIMO setting. There are some instances where the methods used in this paper bears similarities with the methods used for the SISO setting. However, most of the proofs and techniques used here are distinct and considerably more involved than those used in \cite{Wu2006c}. In the SISO setting, it is possible to analyze the model for specific magnitudes of channels. This is not possible for the MIMO setting. We list some of the crucial differences between the methods used in this paper and the methods that have been used under the SISO setting.
\begin{enumerate}
\item In \cite{Wu2006c}, the authors obtain the outer bound using conditional entropy inequality. This method cannot be extended to the MIMO setting.
\item We obtain the outer bound through a series of channel transformations. Although the channel transformations are similar in spirit to those in \cite{Jovicic2006}, the actual transformations used are significantly different both in nature and in the mathematical proofs that accompany them. In \cite{Jovicic2006}, the authors reduce the channel to a broadcast channel where the combined transmitters have individual power constraints and the cognitive receiver has the message of the licensed user provided to it by a genie. The capacity region for such a variation of broadcast channel is not known in general. The authors solve for the capacity region of the broadcast channel using aligned channel techniques. On the other hand, we reduce the MIMO cognitive channel to a broadcast channel with sum power constraint and whose capacity region is now known \cite{Weingarten2006}\cite{Mohseni2006}\cite{TieLiuSubmitted}. We then use optimization techniques to compare the achievable scheme with the outer bound.
\end{enumerate}

\subsection{Main Contributions} In this paper, our main contributions include:

1. We find an achievable region for the Gaussian MIMO cognitive
channel (MCC) in a fashion analogous to
\cite{Devroye2006a}\cite{Jovicic2006}\cite{Wu2006c}.

2. We find an outer bound on the capacity region of the MCC.

3. We show that, under certain conditions (that depend on the channel parameters), the outer bound is tight for a portion of the capacity region boundary, including points corresponding to the sum-capacity of the
channel. Combining the two above, we characterize the sum capacity of this
channel and a portion of its entire capacity region under certain conditions.

\subsection{Organization}
The rest of the paper is organized as follows. We describe the
notations and system model in Section \ref{sec : system model}. The
main results are presented in Section \ref{sec : main results}. In
Section \ref{sec : achievability}, we present an achievable region
for the Gaussian MIMO cognitive channel (MCC). An outer bound on the
capacity region is shown in Section \ref{sec : converse}. The
optimality of the achievable region for a portion of the capacity
region (under certain conditions) is shown in Section \ref{sec : optimality}. Numerical results are provided in Section \ref{sec : numerical results}. We conclude in
Section \ref{sec : conclusions}.

\section{System Model and Notation}\label{sec : system model}
Throughout the paper, we use boldface letters to denote vectors and
matrices. $|\mathbf{A}|$ denotes the determinant of matrix
$\mathbf{A}$, while $\mathrm{Tr}(\mathbf{A})$ denotes its trace. For any
general matrix or vector $\mathbf{X}$, $\mathbf{X}^{\dag}$ denotes
its conjugate transpose. $\mathbf{I}_n$ denotes the $n \times n$
identity matrix. $\mathbf{X}^n$ denotes the row vector $(X(1), X(2),
\ldots, X(n))$, where $X(i), i = 1,2,\ldots,n$ can be vectors or
scalars. The notation $\mathbf{H} \succeq \mathbf{0}$ is used to
denote that a square matrix $\mathbf{H}$ is positive semidefinite. Finally, if $\mathbf{S}$ is a set, then $\mathrm{Cl}(\mathbf{S})$ and $\mathrm{Co}(\mathbf{S})$ denote the closure and convex hull of $\mathbf{S}$ respectively.

We consider a MIMO cognitive channel shown in Figure 1. Let $n_{p,t}$ and $n_{p,r}$ denote the number of
transmitter and receiver antennas respectively for the licensed
user. Similarly, $n_{c,t}$ and $n_{c,r}$ denotes the number of
transmitter and receiver antennas for the cognitive user.

\begin{figure}[hbtp]\label{fig_system_model}
\centering
\includegraphics[width = 3in]{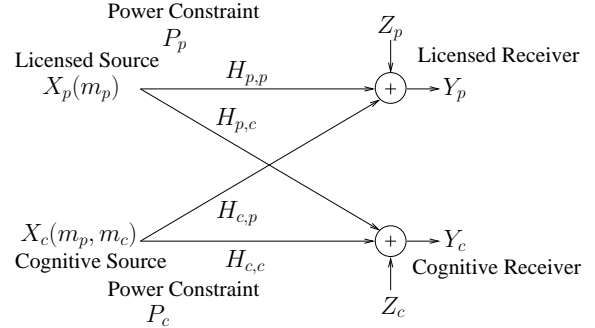}
\caption{MIMO Cognitive Radio System Model}
\end{figure}

The licensed user has message $m_p  \in \{1,2,\ldots,2^{nR_p}\}$
intended for the licensed receiver. The cognitive user has message
$m_c \in \{1,2,\ldots, 2^{nR_c}\}$ intended for the cognitive
receiver as well as the message $m_p$ of the licensed user.

The primary user encodes the message $m_p$ into $\mathbf{X_p}^n$.
Here, $\mathbf{X_p}(i)$ is a $n_{p,t}$ length complex vector. The
cognitive transmitter determines its codeword $\mathbf{X_c}^n$ as a
function of both $m_p$ and $m_c$. Note that the cognitive
transmitter wishes to communicate both $m_p$ (to the licensed
receiver) and $m_c$ (to the cognitive receiver). The channel gain
matrices are given by $\mathbf{H_{p,p}}, \mathbf{H_{p,c}},
\mathbf{H_{c,p}}$ and $\mathbf{H_{c,c}}$, and are assumed to be
static. It is assumed that the licensed receiver knows
$\mathbf{H_{p,p}}, \mathbf{H_{c,p}}$, the licensed transmitter knows
$\mathbf{H_{p,p}}$. It is also assumed that the cognitive
transmitter knows $\mathbf{H_{c,p}}, \mathbf{H_{p,c}},
\mathbf{H_{c,c}}$ and the cognitive receiver knows
$\mathbf{H_{p,c}}, \mathbf{H_{c,c}}$. The received vectors of the
licensed and cognitive users are denoted by $\mathbf{Y_p}^n$ and
$\mathbf{Y_c}^n$ respectively.

With the above model and notations, we can describe the system at
time slot $i$ by
\begin{eqnarray}
\begin{array}{c}
\mathbf{Y_p}(i) = \mathbf{H_{p,p}}\mathbf{ X_p}(i) + \mathbf{H_{c,p}}\mathbf{ X_c}(i) + \mathbf{Z_p}(i) \\
\mathbf{Y_c}(i) = \mathbf{H_{p,c}}\mathbf{ X_p}(i) + \mathbf{H_{c,c}}\mathbf{ X_c}(i) + \mathbf{Z_c}(i).
\end{array}
\end{eqnarray}

The additive noise at the primary and secondary receivers is denoted
by $\mathbf{Z_p}^n$ and $\mathbf{Z_c}^n$ respectively. The noise
vectors $\mathbf{Z_p}^n$ and $\mathbf{Z_c}^n$ are Gaussian and are
assumed to be i.i.d. across symbol times and distributed according
to $\mathcal{N}(0, \mathbf{I_{n_{p,r}}})$ and $\mathcal{N}(0,
\mathbf{I_{n_{c,r}}})$ respectively. The correlation between
$\mathbf{Z_p}^n$ and $\mathbf{Z_c}^n$ is assumed to be arbitrary.
This correlation does not impact the capacity region of the system
as the licensed and the cognitive decoders do not co-operate with
each other. \footnote{A proof of this can be obtained using steps almost exactly identical to those for the broadcast channel in \cite[Exercise 15.10]{ThomasCover}}

We denote the covariance of the codewords of the licensed and
cognitive transmitters at time $i$ by $\mathbf{\Sigma_p}(i)$ and
$\mathbf{\Sigma_c}(i)$ respectively. Then, the transmitters are
constrained by the following transmit power constraints.
\begin{eqnarray} \label{eqn : power constraint}
\begin{array}{c}
\sum_{i=1}^n \mathrm{Tr}(\mathbf{\Sigma_p}(i)) \leq n P_p \\
\sum_{i=1}^n \mathrm{Tr}(\mathbf{\Sigma_c}(i)) \leq n P_c.
\end{array}
\end{eqnarray}

A rate pair $(R_p, R_c)$ is said to be achievable if
\begin{enumerate}
\item there exists a sequence of encoding functions for the licensed
and cognitive users $E_p^n : \{1,\ldots,2^{nR_p}\} \rightarrow
\mathbf{X_p}^n$ and $E_c^n : \{1,\ldots,2^{nR_p}\}\times
\{1,\ldots,2^{nR_c}\} \rightarrow \mathbf{X_c}^n$ such that the
codewords satisfy the power constraints given by (\ref{eqn : power
constraint}),
\item there exists decoding rules $D_p^n : \mathbf{Y_p}^n \rightarrow
\{1,\ldots,2^{nR_p}\}$ and $D_c^n : \mathbf{Y_c}^n \rightarrow
\{1,\ldots,2^{nR_c}\}$ such that the average probability of decoding
error is arbitrarily small for suitably large values of $n$.
\end{enumerate}
The capacity region of the Gaussian MIMO cognitive channel is the
set of all achievable rate pairs $(R_p, R_c)$ and is denoted by $\mathcal{C}_{MCC}$.

\section{Main Results}\label{sec : main results}
In this section, we describe the main results of the paper. Let $\mathbf{G} =\left[ \mathbf{H_{p,p}}\ \ \mathbf{H_{c,p}}\right]$. Let
$\mathcal{R}_{ach}$ denote the set described by
\begin{eqnarray}
\left\{\begin{array}{l} \bigg((R_p, R_c), \mathbf{\Sigma_p}, \mathbf{\Sigma_{c,p}},  \mathbf{\Sigma_{c,c}}, \mathbf{Q}\bigg) : \vspace{0.15cm}\\
R_p \geq 0, R_c \geq 0, \mathbf{\Sigma_p} \succeq \mathbf{0},
\mathbf{\Sigma_{c,p}} \succeq \mathbf{0}, \mathbf{\Sigma_{c,c}} \succeq \mathbf{0} \vspace{0.15cm}\\
R_p \leq \log\left| \mathbf{I} + \mathbf{G} \mathbf{\Sigma_{p,net}} \mathbf{G^{\dagger}} + \mathbf{H_{c,p}} \mathbf{\Sigma_{c,c}} \mathbf{H_{c,p}^{\dagger}}\right|\\
\qquad\quad - \log\left| \mathbf{I} + \mathbf{H_{c,p}\Sigma_{c,c}H_{c,p}^{\dagger}}\right|\vspace{0.15cm}\\
R_c \leq \log\left| \mathbf{I} + \mathbf{H_{c,c} \Sigma_{c,c} H_{c,c}^{\dagger}}\right|\vspace{0.15cm}\\
\mathbf{\Sigma_{p,net}} = \left(\begin{array}{ll} \mathbf{\Sigma_p} & \mathbf{Q} \\\\ \mathbf{Q^{\dagger}}  & \mathbf{\Sigma_{c,p}}\end{array}\right) \succeq \mathbf{0}, \vspace{0.15cm}\\
\mathrm{Tr}(\mathbf{\Sigma_p}) \leq P_p,\ \  \mathrm{Tr}(\mathbf{\Sigma_{c,p}} + \mathbf{\Sigma_{c,c}}) \leq P_c\end{array}\right\}.
\end{eqnarray}
In this setting, $\mathbf{\Sigma_{p, net}}$ is a $(n_{p,t} + n_{c,t}) \times (n_{p,t} + n_{c,t})$ covariance matrix while $\mathbf{\Sigma_{c,c}}$ is a $n_{c,t} \times n_{c,t}$ covariance matrix. $\mathbf{\Sigma_p}$ and $\mathbf{\Sigma_{c,p}}$ represent principal submatrices of $\mathbf{\Sigma_{p, net}}$ of dimensions $n_{p,t} \times n_{p,t}$ and $n_{c,t} \times n_{c,t}$ respectively. The covariances matrices $\mathbf{\Sigma_p}$, $\mathbf{\Sigma_{c,p}}$ and $\mathbf{\Sigma_{c,c}}$ determine the power constraints of the system.

Let $\mathcal{R}_{in}$ denote the closure of the convex hull of the set of rate pairs described by
\begin{eqnarray}\label{eqn : achievable region}
\left\{\begin{array}{c}(R_p, R_c) : \exists\ \mathbf{\Sigma_p},
\mathbf{\Sigma_{c,p}}, \mathbf{\Sigma_{c,c}}, \mathbf{Q},
\textrm{ and }\qquad\\
\qquad\bigg((R_p, R_c), \mathbf{\Sigma_p},
\mathbf{\Sigma_{c,p}}, \mathbf{\Sigma_{c,c}}, \mathbf{Q}\bigg) \in
\mathcal{R}_{ach}\end{array}\right\}.
\end{eqnarray}
\begin{thm} \label{thm : achievability} The capacity region of the MCC, $\mathcal{C}_{MCC}$ satisfies
\begin{equation}
\mathcal{R}_{in} \subseteq \mathcal{C}_{MCC}.
\end{equation}
\end{thm}
The proof of the theorem is given in Section \ref{sec :
achievability}. The coding strategy is based on Costa's dirty paper
coding \cite{Costa1983}\cite{Yu2004}.

We now describe an outer bound on the capacity region of the MIMO
cognitive channel. Let $\alpha > 0$, $\mathbf{G_{\alpha}} =
\left[\mathbf{H_{p,p}}\ \ \frac{\mathbf{H_{c,p}}}{\sqrt{\alpha}}\right]$ and $\overline{\mathbf{K}} = \left[\begin{array}{cc}\mathbf{H_{p,p}} &
\mathbf{H_{c,p}}/\sqrt{\alpha} \\ \mathbf{0} &
\mathbf{H_{c,c}}/\sqrt{\alpha}\end{array}\right]$. Let $\mathbf{\Sigma_z}$  be a covariance matrix of dimensions $(n_{p,r} + n_{c,r}) \times (n_{p,r} + n_{c,r})$ and of the form
\begin{equation}\label{eqn : sigma_z}
\mathbf{\Sigma_z} = \left[\begin{array}{cc} \mathbf{I}_{n_{p,r}} & \mathbf{Q_z} \\ \mathbf{Q_z}^{\dagger} & \mathbf{I}_{n_{c,r}} \end{array}\right].
\end{equation}
Here, $\mathbf{Q_z}$ is a $n_{p,r} \times n_{c,r}$ matrix that makes $\mathbf{\Sigma_z}$ positive semidefinite.
Let $\mathcal{R}_{conv}^{\alpha, \mathbf{\Sigma_z}}$ denote the set described by
\begin{align}\label{eqn : describe converse region}
\left\{\begin{array}{l}\bigg((R_p, R_c), \mathbf{Q_p}, \mathbf{Q_c}\bigg) : \vspace{0.15cm}\\
R_p \geq 0, R_c \geq 0, \mathbf{Q_p} \succeq \mathbf{0}, \mathbf{Q_c} \succeq \mathbf{0}\vspace{0.15cm}\\
R_p \leq \log\left| \mathbf{I} + \mathbf{G_{\alpha}} \mathbf{Q_p} \mathbf{G_{\alpha}^{\dagger}} + \mathbf{G_{\alpha}} \mathbf{Q_c} \mathbf{G_{\alpha}^{\dagger}}\right| \\
\qquad\quad - \log\left| \mathbf{I} + \mathbf{G_{\alpha}} \mathbf{Q_c} \mathbf{G_{\alpha}^{\dagger}}\right| \vspace{0.15cm}\\
R_c \leq \log\left| \mathbf{\Sigma_z} + \overline{\mathbf{K}} \mathbf{Q_c} \overline{\mathbf{K^{\dagger}}}\right| - \log\left|\mathbf{\Sigma_z}\right|\vspace{0.15cm}\\
\mathrm{Tr}(\mathbf{Q_p}) + \mathrm{Tr}(\mathbf{Q_c}) \leq P_p + \alpha P_c \end{array}\right\}.
\end{align}

Let $\mathcal{R}_{out}^{\alpha, \mathbf{\Sigma_z}}$ denote the closure of the convex hull of the set of rate pairs
described by
\begin{align}\label{eqn : converse region}
\left\{(R_p, R_c) : \exists
\mathbf{Q_p}, \mathbf{Q_c} \succeq \mathbf{0},\  ((R_p, R_c), \mathbf{Q_p}, \mathbf{Q_c}) \in
\mathcal{R}_{conv}^{\alpha, \mathbf{\Sigma_z}}\right\}.
\end{align}
Also, let $\mathcal{R}_{out}$ be represented as
\begin{equation}
\mathcal{R}_{out} = \bigcap_{\Sigma_z}\bigcap_{\alpha>0}\mathcal{R}_{out}^{\alpha}.
\end{equation}
Then, the next theorem describes an
outer bound on the capacity region of the MCC.
\begin{thm} \label{thm : converse}
The capacity region of the MCC, $\mathcal{C}_{MCC}$ satisfies
\begin{eqnarray}
\mathcal{C}_{MCC} &\subseteq & \mathcal{R}_{out}^{\alpha, \mathbf{\Sigma_z}}, \forall \alpha > 0 , \mathbf{\Sigma_z}\nonumber\\
\mathcal{C}_{MCC} &\subseteq & \mathcal{R}_{out}.
\end{eqnarray}
\end{thm}
The proof is given in Section \ref{sec : converse} and proceeds by a
series of channel transformations. Each channel transformation results in
a new channel whose capacity region is in general a superset (outer bound) of
the capacity region of the preceding channel.

Let $\mathrm{BC}(\mathbf{H_1}, \mathbf{H_2}, P)$ denote a two user MIMO broadcast channel with channel matrices given by $\mathbf{H_1}$ and $\mathbf{H_2}$ and with a transmitter power constraint of $P$. Let $\mathcal{C}_{BC}^{H_1, H_2, P}$ denote the capacity region of $\mathrm{BC}(\mathbf{H_1}, \mathbf{H_2}, P)$.

Let $\mathcal{R}_{part,conv}^{\alpha}$ denote the set described by
\begin{align}\label{eqn : describe partial converse region}
\left\{\begin{array}{l}\bigg((R_p, R_c), \mathbf{Q_p}, \mathbf{\Sigma_{c,c}}\bigg) : \vspace{0.15cm}\\
R_p \geq 0, R_c \geq 0, \mathbf{Q_p} \succeq \mathbf{0}, \mathbf{\Sigma_{c,c}} \succeq \mathbf{0},\vspace{0.15cm}\\
R_p \leq \log\left| \mathbf{I} + \mathbf{G_{\alpha}} \mathbf{Q_p} \mathbf{G_{\alpha}^{\dagger}} + \frac{1}{\alpha}\mathbf{H_{c,p}} \mathbf{\Sigma_{c,c}} \mathbf{\Sigma_{c,p}^{\dagger}}\right|\vspace{0.05cm}\\
\qquad\quad - \log\left| \mathbf{I} + \frac{1}{\alpha}\mathbf{H_{c,p}} \mathbf{\Sigma_{c,c}} \mathbf{H_{c,p}^{\dagger}}\right| \vspace{0.15cm}\\
R_c \leq \log\left| \mathbf{I} + \frac{1}{\alpha}\mathbf{H_{c,c}} \mathbf{\Sigma_{c,c}} \mathbf{H_{c,c}^{\dagger}}\right| \vspace{0.15cm}\\
\mathrm{Tr}(\mathbf{Q_p}) + \mathrm{Tr}(\mathbf{\Sigma_{c,c}}) \leq P_p + \alpha P_c \end{array}\right\}.
\end{align}
We let $\mathcal{R}_{part, out}^{\alpha}$ to denote the closure of the convex hull of the set of rate pairs described by
\begin{align}\label{eqn : partial converse region}
\left\{\begin{array}{l}(R_p, R_c) : \exists
\mathbf{Q_p}, \mathbf{\Sigma_{c,c}} \succeq \mathbf{0} \textrm{ and }\\
\qquad\qquad ((R_p, R_c), \mathbf{Q_p}, \mathbf{\Sigma_{c,c}}) \in
\mathcal{R}_{part, conv}^{\alpha}\end{array}\right\}.
\end{align}

Let $\mathbf{K} = \left[\mathbf{0}\ \mathbf{H_{c,c}}/\sqrt{\alpha}\right]$. We show that if the boundary of the rate region described by $\mathcal{R}_{part, out}^{\alpha}$ partially meets the boundary of the capacity region of $BC(\mathbf{G}_{\alpha}, \mathbf{K}, P_p + \alpha P_c)$, then the boundary of $\mathcal{R}_{part, out}^{\alpha}$ partially meets the boundary of the rate region described by $\mathcal{R}_{out}^{\alpha, \mathbf{\Sigma_z}}$ in (\ref{eqn : converse region}) for some $\mathbf{\Sigma_z}$. We formally state the result in Theorem \ref{thm : partial converse}. For notational convenience, we will denote the capacity region of $BC(\mathbf{G}_{\alpha}, \mathbf{K}, P_p + \alpha P_c)$ by $\mathcal{C}_{BC}^{\alpha}$.
\begin{thm}\label{thm : partial converse}
Let $\mu \geq 1$ and $\alpha > 0$. If
\begin{equation}\label{eqn : condition}
\max_{(R_p, R_c) \in \mathcal{R}_{part, out}^{\alpha}} \mu R_p + R_c = \max_{(R_p, R_c) \in \mathcal{C}_{BC}^{\alpha}} \mu R_p + R_c,
\end{equation}
then, we have
\begin{equation}
\max_{(R_p, R_c) \in \mathcal{R}_{part, out}^{\alpha}} \mu R_p + R_c = \inf_{\mathbf{\Sigma_z}}\max_{(R_p, R_c) \in \mathcal{R}_{out}^{\alpha, \mathbf{\Sigma_z}}} \mu R_p + R_c.
\end{equation}
\end{thm}
The proof of the theorem is described in Section \ref{sec : converse}. Hence, if the condition (\ref{eqn : condition}) is satisfied, the rate region described by $\mathcal{R}_{part, out}^{\alpha}$ is an outer bound on the capacity region of the MCC in terms of maximizing the $\mu$- sum $\mu R_p + R_c$.

Let $(\hat{R}_p, \hat{R}_c)$ be a point on the boundary of the
capacity region $\mathcal{C}_{MCC}$. Then, there exists a $\mu \geq 0$
such that
\begin{displaymath}
(\hat{R}_p, \hat{R}_c) = \arg \max_{(R_p, R_c) \in
\mathcal{C}_{MCC}} \mu R_p + R_c.
\end{displaymath}
The next theorem shows that if $(R_p, R_c)$ lies on the boundary of the achievable region given by
$\mathcal{R}_{in}$, then $(R_p, R_c)$ lies on the boundary of $\mathcal{R}_{part, out}^{\alpha}$ for some $\alpha > 0$. That is, the theorem describes conditions of optimality of the achievable region $\mathcal{R}_{in}$.
\begin{thm}\label{thm : optimality}
For any $\mu > 0$,
\begin{displaymath}
\max_{(R_p, R_c) \in \mathcal{R}_{in}} \mu R_p + R_c = \inf_{\alpha
> 0} \max_{(R_p, R_c) \in \mathcal{R}_{part, out}^{\alpha}} \mu R_p + R_c.
\end{displaymath}
Also, there exists $\alpha^* \in (0, \infty)$, such that for any $\mu \geq 1$, $(R_{p,\mu}, R_{c,\mu}) = \arg \max_{(R_p, R_c) \in \mathcal{R}_{in}} \mu R_p + R_c$ is a point on the boundary of the
capacity region of the MIMO cognitive channel if
the condition given by (\ref{eqn : condition}) is satisfied for $\alpha^*$.
\end{thm}
The proof of the  theorem is described in Section \ref{sec :
optimality} and is based on optimization techniques.

\section{Achievable Region}\label{sec : achievability}
\begin{proof} \textit{of Theorem \ref{thm : achievability}} : In this section, we
show that the rate region $\mathcal{R}_{in}$ given by (\ref{eqn :
achievable region}) is achievable on the MCC.

Encoding rule for Licensed user $(E_p^n)$ : For every message $m_p
\in \{1,\ldots,2^{nR_p}\}$, the licensed encoder generates a $n$
length codeword $\mathbf{X_p}^n(m_p)$, according to the distribution
$p(\mathbf{X_p}^n) = \Pi_{i=1}^n p(\mathbf{X_p}(i))$, and
$X_p(i)\backsim \mathcal{N}(\mathbf{0}, \mathbf{\Sigma_p})$ such
that $\mathbf{\Sigma_p} \succeq \mathbf{0}$ and
$\mathrm{Tr}(\mathbf{\Sigma_p}) \leq P_p$.

Encoding rule for the cognitive user $(E_c^n) $: The cognitive
encoder acts in two stages. For every message pair $(m_p, m_c)$, the
cognitive encoder first generates a codeword
$\mathbf{X_{c,p}}^n(m_p, m_c)$ for the primary message $m_p$
according to $\Pi_{i=1}^n p(\mathbf{X_{c,p}}(i) | \mathbf{X_p}(i))$,
where $p(\mathbf{X_{c,p}}(i)) \backsim \mathcal{N}(\mathbf{0}, \mathbf{\Sigma_{c,p}})$ and the joint distribution of $(\mathbf{X_p}(i), \mathbf{X_{c,p}}(i))$ is given by
\begin{equation}
p(\mathbf{X_p}(i), \mathbf{X_{c,p}}(i)) \backsim \mathcal{N}\Bigg(\mathbf{0}, \left[\begin{array}{cc} \mathbf{\Sigma_p} & \mathbf{Q} \\ \mathbf{Q}^{\dagger} & \mathbf{\Sigma_{c,c}} \end{array}\right]\Bigg).
\end{equation}
Here, $\mathbf{Q}$ denotes the correlation between $\mathbf{X_p}(i)$ and $\mathbf{X_{c,p}}(i)$. In the second
stage, the cognitive encoder generates $\mathbf{X_{c,c}}^n$ which
encodes message $m_c$. The codeword $\mathbf{X_{c,c}}^n$ is
generated using Costa precoding \cite{Costa1983} by treating
$\mathbf{H_{p,p}}\mathbf{X_p}^n + \mathbf{H_{c,c}}
\mathbf{X_{c,p}}^n$ as non causally known interference. A
characteristic feature of Costa's precoding is that
$\mathbf{X_{c,c}}^n $ is independent of $\mathbf{X_{c,p}}^n$, and
$\mathbf{X_{c,c}}^n$ is distributed as $\Pi_{i=1}^n
p(\mathbf{X_{c,c}}(i))$, where $\mathbf{X_{c,c}}(i) \backsim
\mathcal{N}(\mathbf{0}, \mathbf{\Sigma_{c,c}})$. Note that the
codeword $\mathbf{X_{c,p}}^n$ is used to convey message $m_p$ to the
licensed receiver and the codeword $\mathbf{X_{c,c}}^n$ is used to
convey message $m_c$ to the cognitive receiver. The two codewords
$\mathbf{X_{c,p}}^n$ and $\mathbf{X_{ c,c}}^n$ are superimposed to
form the cognitive codeword $\mathbf{X_c}^n = \mathbf{X_{c,p}}^n +
\mathbf{X_{c,c}}^n$. It is clear that $\mathbf{X_c}^n$ is
distributed as $\Pi_{i=1}^n p(\mathbf{X_c}(i)), \ \mathbf{X_c}(i)
\backsim \mathcal{N}(\mathbf{0}, \mathbf{\Sigma_c})$, where
$\mathbf{\Sigma_c} = \mathbf{\Sigma_{c,p}} + \mathbf{\Sigma_{c,c}}$.
The covariance matrices satisfy the constraints
$\mathbf{\Sigma_{c,p}} \succeq \mathbf{0}, \mathbf{\Sigma_{c,c}}
\succeq \mathbf{0}, \mathrm{Tr}(\mathbf{\Sigma_c}) \leq P_c$.

Decoding rule for the licensed receiver $(D_p^n)$ : The licensed
receiver receives $\mathbf{H_{p,p} X_p}^n + \mathbf{H_{c,p}} (
\mathbf{X_{c,p}}^n + \mathbf{X_{c,c}}^n ) + \mathbf{Z_p}^n$. It
treats $\mathbf{H_{p,p} X_p}^n + \mathbf{H_{c,p}  X_{c,p}}^n$ as the
valid codeword and $\mathbf{H_{c,p} X_{c,c}}^n + \mathbf{Z_p}^n$ as
Gaussian noise. Taking $\mathbf{G} = \left[\mathbf{H_{p,p}}\ \
\mathbf{H_{c,p}}\right]$ and $\mathbf{X_{p,net}}^n =
\left[\begin{array}{c} \mathbf{X_p}^n \\
\mathbf{X_{c,p}}^n\end{array}\right]$,  the received vector at the
licensed receiver is
\begin{equation}
\mathbf{Y_p}^n = \mathbf{G} \mathbf{X_{p,net}}^n + \mathbf{H_{c,p}}
\mathbf{X_{c,c}}^n + \mathbf{Z_p}^n.
\end{equation}
The covariance matrix of $\mathbf{X_{p,net}}$ is denoted by
$\mathbf{\Sigma_{p,net}} = \left[\begin{array}{cc}\mathbf{\Sigma_p}
& \mathbf{Q} \\ \mathbf{Q^{\dagger}} &
\mathbf{\Sigma_{c,p}}\end{array}\right]$, where $\mathbf{Q} =
E[\mathbf{X_p} \mathbf{X_{c,p}^{\dagger}}]$. In this setup, we use steps identical to that used for MIMO channel with colored noise in \cite[Section 9.5]{ThomasCover} to show that, for any $\epsilon > 0$,
there exists a block length $n_1$ so that for any $n \geq n_1$, the licensed decoder can
recover the message $m_p$ with probability of error $ < \epsilon$ if
\begin{equation}
\begin{array}{c}
R_p \leq \log\left| \mathbf{I} + \mathbf{G} \mathbf{\Sigma_{p,net}}
\mathbf{G^{\dagger}} + \mathbf{H_{c,p}} \mathbf{\Sigma_{c,c}}
\mathbf{H_{c,p}^{\dagger}}\right| \vspace{0.1cm}\\
- \log\left| \mathbf{I} +
\mathbf{H_{c,p}} \mathbf{\Sigma_{c,c}} \mathbf{H_{c,p}^{\dagger}}\right|.\end{array}
\end{equation}

Decoding rule for the cognitive user $(D_c^n)$ : The cognitive
decoder is the Costa decoder (with the knowledge of the encoder,
$E_c^n$). The cognitive receiver receives $\mathbf{Y_c}^n =
\mathbf{H_{p,c}}\mathbf{X_p}^n + \mathbf{H_{c,c}}
(\mathbf{X_{c,p}}^n + \mathbf{X_{c,c}}^n) + \mathbf{Z_c}^n$. Here,
the non-causally known interference $\mathbf{H_{p,c}}\mathbf{X_p}^n
+ \mathbf{H_{c,c}}\mathbf{X_{c,p}}^n$ is canceled by the Costa
precoder. To show this formally, we follow steps similar to Eqns (3) to (7)  in \cite{Costa1983}. We get that, for any $\epsilon_2 > 0$, there exists $n_2$ such
that for $n \geq n_2$, the cognitive decoder can recover the message
$m_c$ with probability of error $ < \epsilon_2$ if
\begin{equation}
R_c \leq \log\left| \mathbf{I} + \mathbf{H_{c,c}} \mathbf{\Sigma_{c,c}}
\mathbf{H_{c,c}^{\dagger}}\right|.
\end{equation}

Note that the achievable scheme holds for all possible covariance
matrices $\mathbf{\Sigma_p}, \mathbf{\Sigma_{c,p}},
\mathbf{\Sigma_{c,c}}$ that are positive semidefinite and satisfy
the power constraints $\mathrm{Tr}(\mathbf{\Sigma_p}) \leq P_p,
\mathrm{Tr}(\mathbf{\Sigma_{c,p}} + \mathbf{\Sigma_{c,c}}) \leq P_c$. Hence,
$\mathcal{R}_{in}$, which is the  set of all achievable rate pairs
described by (\ref{eqn : achievable region}), is achievable for any code length $n \geq \max(n_1, n_2)$.
\end{proof}

\section{Outer Bound on the Capacity Region}\label{sec : converse}
In this section, we prove that the rate region described by
$\mathcal{R}_{out}^{\alpha, \mathbf{\Sigma_z}}$ is an outer bound on the capacity
region of the Gaussian MIMO cognitive channel.  The proof proceeds
by a series of channel transformations where each transformation
creates an outer bound on the channel at the previous stage. At the
final stage, we obtain a physically degraded broadcast channel. The
capacity region of this channel is now known \cite{Weingarten2006}\cite{Mohseni2006}\cite{TieLiuSubmitted} and is used as the outer bound for
the capacity region of the MIMO cognitive channel. Figure $2$ depicts the various channel configurations considered,
and the system equations of all the configurations.
$\mathbf{\hat{Z}_p}^n$ shown in Figures $2$c, $2$d and $2$e has the
same distribution as $\mathbf{Z_p}^n$, but has an arbitrary
correlation with $\mathbf{Z_c}^n$.
\begin{figure*}[!th]\label{fig : converse}
\centering
\includegraphics[width = 6 in]{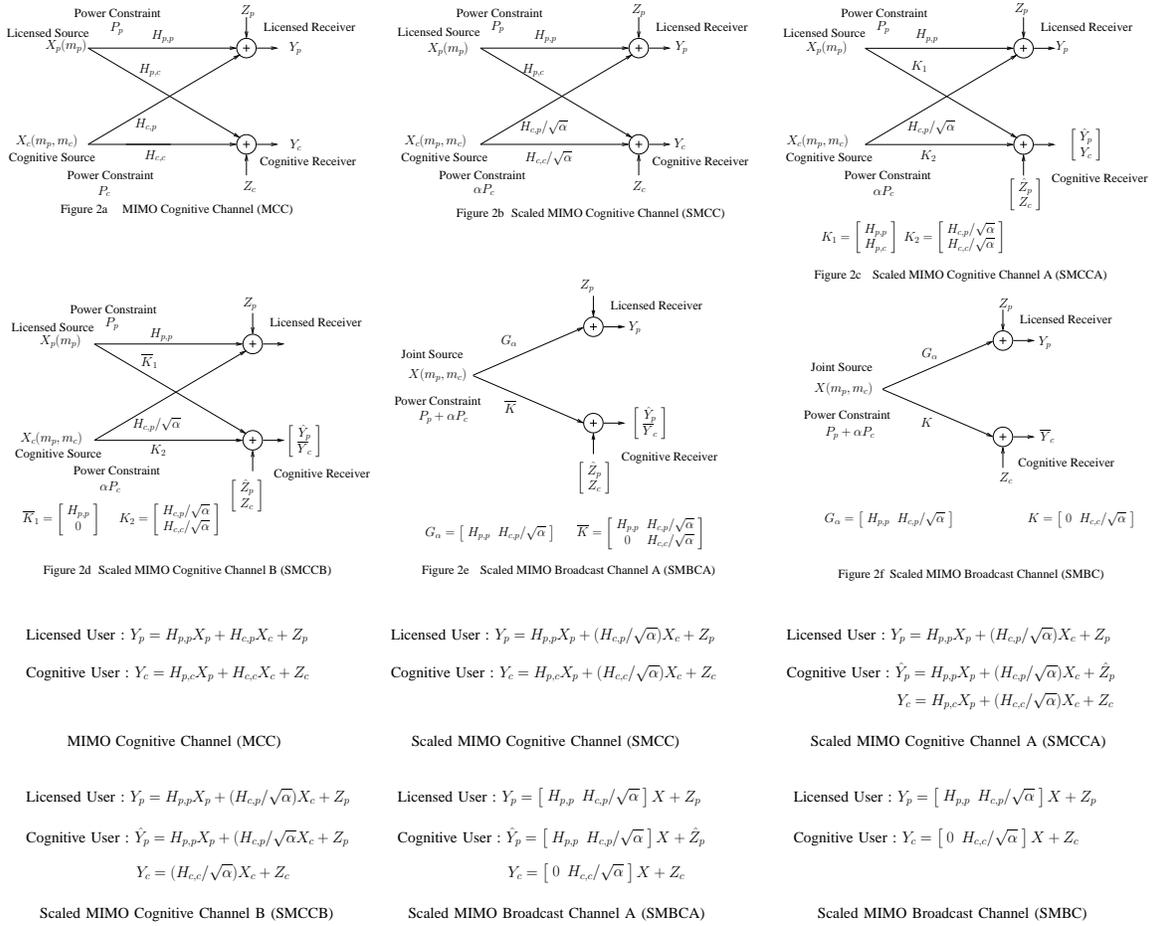}
\caption{Channel Configurations and their System Equations}
\end{figure*}
Before proving Theorem \ref{thm : converse}, we prove the following
lemmas.

Transformation 1 (MIMO Cognitive Channel (MCC) $\rightarrow$ Scaled
MIMO cognitive channel) : The scaled MIMO cognitive channel is
defined in Figure $2\textrm{b}$ and Figure $3$. In this transformation, the
channel matrices $\mathbf{H_{c,p}}$ and $\mathbf{H_{c,c}}$ are
scaled by $1/\sqrt{\alpha}$. Also, the power constraint at the
cognitive transmitter is changed to $\alpha P_c$.

\begin{lem} The capacity region of the MIMO cognitive channel is
equal to the capacity region of the scaled MIMO cognitive
channel (SMCC) for any $0 < \alpha < \infty$. \end{lem}
\begin{figure*}[!th]\label{fig : transformation 1}
\centering
\includegraphics[width = 4.6 in]{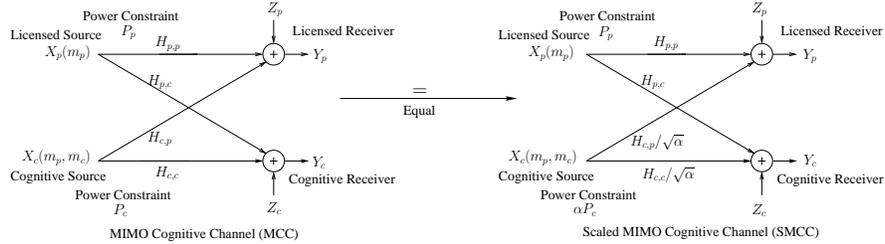}
\caption{Capacity Region of MCC = Capacity Region of SMCC}
\end{figure*}
\begin{proof} :  Let $(R_p, R_c)$ be a rate pair that is achievable on
the MCC. That is, for all $\epsilon_1, \epsilon_2 > 0$, there exists
a $n$ and a sequence of encoder decoder pairs at the licensed and
cognitive transmitter and receiver $(E_p^n : m_p \rightarrow
\mathbf{X_p}^n, D_p^n : \mathbf{Y_p}^n \rightarrow \hat{m}_p, E_c^n
: (m_p, m_c) \rightarrow \mathbf{X_c}^n, D_c^n : \mathbf{Y_c}^n
\rightarrow \hat{m}_c)$ such that the codewords $\mathbf{X_p}^n$ and
$\mathbf{X_c}^n$ satisfy the power constraints given by (\ref{eqn :
power constraint}) and the probability of decoding error is small
$(Pr(m_p \neq \hat{m}_p) \leq \epsilon_1, Pr(m_c \neq \hat{m}_c)
\leq \epsilon_2)$. We use the following encoder decoder pairs at the
licensed and cognitive transmitters and receivers of the scaled MIMO
cognitive channel. $E_p^n : m_p \rightarrow \mathbf{X_p}^n,\  D_p^n
: \mathbf{Y_p}^n \rightarrow \hat{m}_p,\  E_c^n : (m_p, m_c)
\rightarrow \sqrt{\alpha}\mathbf{X_c}^n,\  D_c^n : \mathbf{Y_c}^n
\rightarrow \hat{m}_c$. It follows that using these encoder
and decoder pairs, the licensed and cognitive codewords satisfy the
new power constraints of $P_p$ and $\alpha P_c$ respectively. Also,
the system equation is the same as that of the MCC and $Pr(m_p \neq
\hat{m}_p) \leq \epsilon_1$ and $Pr(m_c \neq \hat{m}_c) \leq
\epsilon_2$. Hence, the rate pair $(R_p,R_c)$ is achievable on the
scaled MIMO cognitive channel. Hence, the capacity region of the
SMCC is a superset of the capacity region of the MCC.

Similarly, we can also establish this in the other direction, namely we can treat the MCC as the scaled version of the SMCC
(scaling by $1/\alpha$). Therefore, it can be shown that the
capacity region of the MCC is a superset of the capacity region
of the SMCC.

Hence, the capacity region of the MCC is equal to the capacity
region of the SMCC.
\end{proof}

Transformation 2 (scaled MIMO cognitive channel (SMCC) $\rightarrow$
scaled MIMO cognitive channel A (SMCCA)) : The scaled MIMO cognitive
channel A (SMCCA) is described in Figure $2$c and Figure $4$. In this
transformation, we provide a modified version of $\mathbf{Y_p}^n$, which is $\mathbf{\hat{Y}_p}^n$ to the cognitive receiver. $\mathbf{\hat{Y}_p}^n$ is corrupted by noise $\mathbf{\hat{Z}_p}^n$, which has the same
probability distribution as that of $\mathbf{Z_p}^n$ (i.e., complex Gaussian with zero mean and identity covariance matrix), but is permitted to be correlated with $\mathbf{Z_p}^n$ or $\mathbf{Z_c}^n$. In fact, we assume that the joint probability distribution of $(\mathbf{\hat{Z}_p}(i), \mathbf{Z_c}(i))$ is given by
\begin{equation}
p(\mathbf{\hat{Z}_p}(i), \mathbf{Z_c}(i)) = \mathcal{N}(0, \mathbf{\Sigma_z}),
\end{equation}
where $\mathbf{\Sigma_z}$ has the form given by (\ref{eqn : sigma_z}). The
received vector $\mathbf{\hat{Y}_p}^n$ is made available to the
cognitive receiver by transforming the channel matrices
$\mathbf{H_{p,c}}$ and $\mathbf{H_{c,c}}/\sqrt{\alpha}$ to $K_1 = \left[\begin{array}{c}\mathbf{H_{p,p}} \\
\mathbf{H_{p,c}} \end{array}\right]$ and $K_2 =
\left[\begin{array}{c} \mathbf{H_{c,p}}/\sqrt{\alpha} \\
\mathbf{H_{c,c}}/\sqrt{\alpha} \end{array}\right]$ respectively.
Hence, the received vector at the cognitive receiver is
$\left[\begin{array}{c}\mathbf{\hat{Y}_p}^n \\
\mathbf{Y_c}^n\end{array}\right]$.
\begin{lem} The capacity region of the scaled MIMO cognitive channel A (SMCCA) is a superset of the capacity region of the scaled MIMO cognitive channel (SMCC). \end{lem}
\begin{figure*}[!th]\label{fig : transformation 2}
\centering
\includegraphics[width=4.6in]{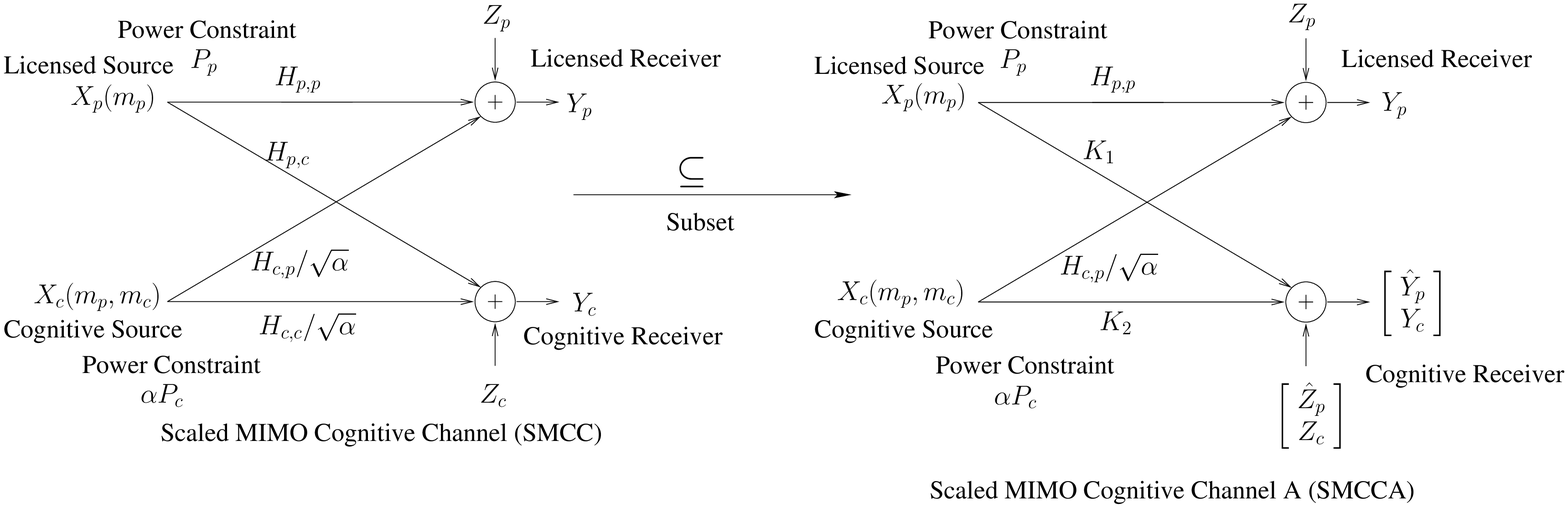}
\caption{Capacity Region of SMCC $\subseteq$ Capacity Region of SMCCA}
\end{figure*}
\begin{proof} : Let the rate pair $(R_p, R_c)$ be achievable on the SMCC. That is, for all $\epsilon_1, \epsilon_2 > 0$, there exists
a $n$ and a sequence of encoder decoder pairs at the licensed and
cognitive transmitter and receiver $(E_p^n : m_p \rightarrow
\mathbf{X_p}^n, D_p^n : \mathbf{Y_p}^n \rightarrow \hat{m}_p, E_c^n
: (m_p, m_c) \rightarrow \mathbf{X_c}^n, D_c^n : \mathbf{Y_c}^n
\rightarrow \hat{m}_c)$ such that the codewords $\mathbf{X_p}^n$ and
$\mathbf{X_c}^n$ satisfy the power constraints and the probability of decoding error is small
$(Pr(m_p \neq \hat{m}_p) \leq \epsilon_1, Pr(m_c \neq \hat{m}_c)
\leq \epsilon_2)$. In the SMCCA, we can use the same encoder decoder pair $E_p^n$ and $D_p^n$ at the licensed transmitter and receiver to achieve a rate $R_p$ with probability of decoding error $< \epsilon_1$. Also, by ignoring the received vector $\mathbf{\hat{Y}_p}^n$ at the cognitive receiver, we can use $E_c^n$ and $D_c^n$ at the cognitive transmitters and receivers to achieve a rate $R_c$ with the decoding probability of error $< \epsilon_2$. Hence, the rate pair $(R_p, R_c)$ is achievable on the scaled MIMO cognitive channel A (SMCCA). Therefore, the capacity region of the SMCCA is a superset of the capacity region of the SMCC.
\end{proof}

Transformation 3 (scaled MIMO cognitive channel A (SMCCA)
$\rightarrow$ scaled MIMO cognitive channel B (SMCCB) ) : The scaled
MIMO cognitive channel (B) is described in Figure $2$d and Figure $5$.
The channel matrix from the licensed transmitter to the cognitive
receiver is modified from $\mathbf{K_1} = \left[\begin{array}{c}\mathbf{H_{p,p}} \\
\mathbf{H_{p,c}} \end{array}\right]$ to $\overline{\mathbf{K_1}} =
\left[\begin{array}{c} \mathbf{H_{p,p}} \\ \mathbf{0}
\end{array}\right]$. Hence, the received vector at the cognitive
receiver is given by $\left[\begin{array}{c} \mathbf{\hat{Y}_p}^n \\
\overline{\mathbf{Y_c}^n}\end{array}\right]$ where
$\overline{\mathbf{Y_c}^n} = \frac{\mathbf{H_{c,c}}}{\sqrt{\alpha}}
\mathbf{X_c}^n + \mathbf{Z_c}^n$. The intuition behind the
transformation is to remove the original interference caused by the
licensed transmitter to the cognitive receiver.
\begin{lem} The capacity region of the scaled MIMO cognitive channel B (SMCCB) is equal
to the capacity region of the scaled MIMO cognitive channel A (SMCCA). \end{lem}
\begin{figure*}[!th]\label{fig : transformation 3}
\centering
\includegraphics[width=4.6in]{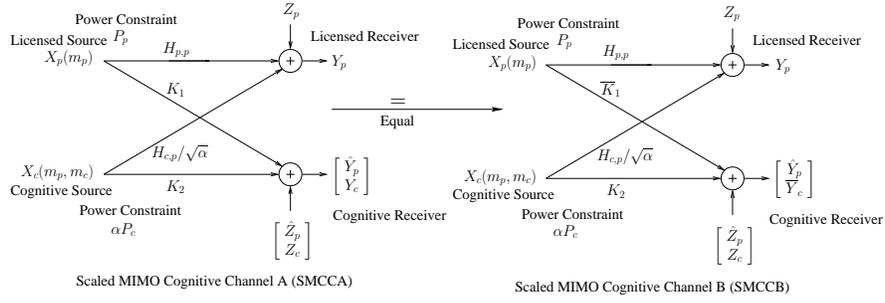}
\caption{Capacity Region of SMCCA = Capacity Region of SMCCB}
\end{figure*}
\begin{proof} :
Let the rate pair $(R_p, R_c)$ be achievable on the SMCCA. This
implies that for every $\epsilon_1, \epsilon_2 > 0$, there exists
encoder-decoder pair for the licensed user $(E_p^n(\epsilon_1),
D_p^n(\epsilon_1))$ and for the cognitive user $(E_c^n(\epsilon_2),
D_c^n(\epsilon_2))$ such that the probability of decoding error is
less than $\epsilon_1$ and $\epsilon_2$ respectively for the
licensed and cognitive user. Let $\delta_1, \delta_2 \in (0,1)$. In
SMCCB, the licensed user can employ $E_p^n(\min (\delta_1/2,
\delta_2/2)),$ $D_p^n(\min (\delta_1/2, \delta_2/2))$ to decode $m_p$
with a probability of error $\leq \delta_1/2 < \delta_1$. The
cognitive receiver uses $E_p^n(\min (\delta_1/2, \delta_2/2)), $
$D_p^n(\min (\delta_1/2, \delta_2/2))$ on $\mathbf{\hat{Y}_p}^n$ to
obtain $m_p$ with probability of error $\leq \delta_1/2$. The
cognitive receiver can now construct $\mathbf{X_p}^n$ and hence
$\mathbf{H_{p,c}} \mathbf{X_p}^n$. Thus, the cognitive receiver
recovers $\mathbf{Y_c}^n = \mathbf{H_{p,c}} \mathbf{X_p}^n +
\frac{\mathbf{H_{c,c}}}{\sqrt{\alpha}} \mathbf{X_{c,c}^n} +
\mathbf{Z_c}^n$. Now, it uses, $E_c^n(\delta_2/2), D_c^n(\delta_2/2)$
to obtain $m_c$ with probability of error $\leq \delta_2/2$.
Clearly, the probability of error in recovering $m_c$ is less than
$\delta_2$. Hence, the rate pair $(R_p, R_c)$ is achievable on
SMCCB. Therefore, the capacity region of SMCCB is a superset of the capacity
region of SMCCA.

Let the rate pair $(R_p, R_c)$ be achievable on SMCCB. Then, for every $\epsilon_1, \epsilon_2 > 0$, there exists
encoder-decoder pair for the licensed user $(E_p^n(\epsilon_1),
D_p^n(\epsilon_1))$ and for the cognitive user $(E_c^n(\epsilon_2),
D_c^n(\epsilon_2))$ such that the probability of decoding error is
less than $\epsilon_1$ and $\epsilon_2$ respectively for the
licensed and cognitive user. Let $\delta_1, \delta_2 > 0$. In SMCCA, the licensed user can employ $E_p^n(\min (\delta_1/2,
\delta_2/2)), D_p^n(\min (\delta_1/2, \delta_2/2))$ to decode $m_p$
with a probability of error $\leq \delta_1/2 < \delta_1$. The cognitive user employs $E_p^n(\min (\delta_1/2, \delta_2/2)), D_p^n(\min (\delta_1/2, \delta_2/2))$ on $\mathbf{\hat{Y}_p}^n$ to
obtain $m_p$ with probability of error $\leq \delta_2/2$. The cognitive receiver can now construct $\mathbf{X_p}^n$ and hence $\mathbf{H_{p,c}}\mathbf{X_p}^n$. Hence, the cognitive receiver subtracts $\mathbf{H_{p,c}}\mathbf{X_p}^n$ from $\mathbf{Y_c}^n$ to obtain $\mathbf{\overline{Y}_c}^n$. The cognitive receiver can now use $E_c^n(\delta_2/2), D_c^n(\delta_2/2)$ to obtain $m_c$ with probability of error $< \delta_2$. Thus, the rate pair $(R_p, R_c)$ is achievable on SMCCA.

Therefore, the capacity region of the SMCCA is equal to the capacity region of the SMCCB.
\end{proof}

Transformation 4 (scaled MIMO cognitive channel (B) $\rightarrow$
scaled MIMO broadcast channel A (SMBCA)): The scaled MIMO broadcast
channel A (SMBCA) is depicted in Figure $2$e and Figure $6$. We let
the two transmitters to co-operate and transform it into a broadcast
channel with a sum power constraint of $P_p + \alpha P_c$. The new
channel matrices from the combined transmitters to the licensed and
cognitive receivers are given by $\mathbf{G_{\alpha}} =
\left[\begin{array}{cc} \mathbf{H_{p,p}} &
\mathbf{H_{c,p}}/\sqrt{\alpha} \end{array}\right]$ and
$\overline{\mathbf{K}} = \left[\begin{array}{cc}\mathbf{H_{p,p}} &
\mathbf{H_{c,p}}/\sqrt{\alpha} \\ \mathbf{0} &
\mathbf{H_{c,c}}/\sqrt{\alpha}\end{array}\right]$ respectively.
\begin{lem} The capacity region of the scaled MIMO broadcast channel A (SMBCA) is a superset of the capacity region of scaled MIMO cognitive channel B (SMCCB).\end{lem}
\begin{figure*}[!th]\label{fig : transformation 4}
\centering
\includegraphics[width=4.6in]{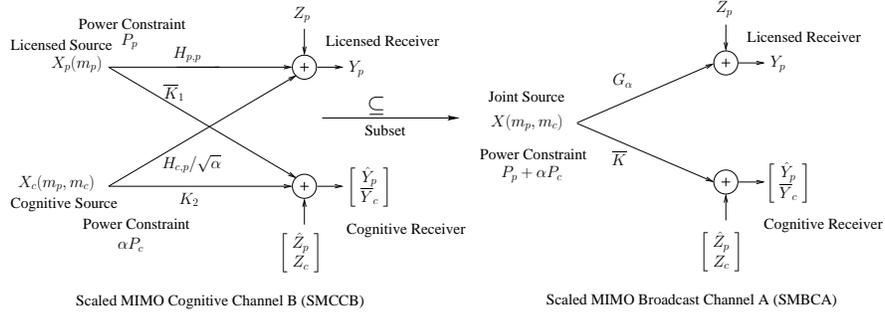}
\caption{Capacity Region of SMCCB $\subseteq$ Capacity Region of SMBCA}
\end{figure*}
\begin{proof} :
Let the rate pair $(R_p, R_c)$ be achievable on the SMCCB. In the SMBCA, using no collaboration between the two transmitters and using separate power constraints of $P_p$ and $\alpha P_c$ respectively, we reduce the SMBCA to the SMCCB. Hence, the rate pair $(R_p, R_c)$ is achievable on the SMBCA. Therefore, the capacity region of the SMBCA is a superset of the capacity region of the SMCCB.
\end{proof}

We have showed that for any $\alpha > 0$, $\mathcal{C}_{MCC} =
\mathcal{C}_{SMCC} \subseteq \mathcal{C}_{SMCCA} =
\mathcal{C}_{SMCCB} \subseteq \mathcal{C}_{SMBCA}$. Hence, the capacity region of the scaled MIMO broadcast channel A (SMBCA) is a superset of the capacity region of the MIMO cognitive channel (MCC).

\begin{proof} \textit{of Theorem \ref{thm : converse}} : In the SMBCA, let $\mathbf{Q_p}$ denote the covariance matrix of the codeword for the licensed user and let $\mathbf{Q_c}$ denote the covariance matrix for the cognitive user. The SMBCA is a physically degraded broadcast channel. Hence, the capacity region of the SMBCA (as given by \cite{Weingarten2006}) denoted by $\mathcal{C}_{SMBCA}$ is the closure of the convex hull of the set of rate pairs described by
\begin{align}
\left\{\begin{array}{l} (R_p, R_c) : R_p \geq 0, R_c \geq 0 \vspace{0.15cm}\\
R_p \leq \log\left| \mathbf{I} + \mathbf{G_{\alpha}} \mathbf{Q_p} \mathbf{G_{\alpha}^{\dagger}} + \mathbf{G_{\alpha}} \mathbf{Q_c} \mathbf{G_{\alpha}^{\dagger}}\right| \\
\qquad\quad - \log\left| \mathbf{I} + \mathbf{G_{\alpha}} \mathbf{Q_c} \mathbf{G_{\alpha}^{\dagger}}\right| \vspace{0.15cm}\\
R_c \leq \log\left| \mathbf{\Sigma_z} + \overline{\mathbf{K}} \mathbf{Q_c} \overline{\mathbf{K^{\dagger}}}\right| - \log\left|\mathbf{\Sigma_z}\right|\vspace{0.15cm}\\
\forall \mathbf{Q_p} \succeq \mathbf{0}, \mathbf{Q_c} \succeq \mathbf{0}\vspace{0.15cm}\\
\mathrm{Tr}(\mathbf{Q_p}) + \mathrm{Tr}(\mathbf{Q_c}) \leq P_p + \alpha P_c \end{array}\right\}.
\end{align}

Also, this is the outer bound of the MCC. Hence, $\mathcal{R}_{out}^{\alpha, \mathbf{\Sigma_z}}$ described by (\ref{eqn : converse region}) is an outer bound on the capacity region of the MCC. Hence, $\mathcal{C}_{MCC} \subseteq \mathcal{R}_{out}^{\alpha, \mathbf{\Sigma_z}}$. Also, $\mathcal{C}_{MCC} \subseteq \mathcal{R}_{out}$, where $\mathcal{R}_{out}$ is described in (9).
\end{proof}

Transformation 5 (scaled MIMO broadcast channel A (SMBCA)
$\rightarrow$ scaled MIMO broadcast channel (SMBC)) : The scaled
MIMO broadcast channel (SMBC) is depicted in Figure $2$f and Figure
$7$. We change the received vector at the cognitive receiver from
$\left[\begin{array}{c} \mathbf{\hat{Y}_p}^n \\
\overline{\mathbf{Y_c}^n}
\end{array}\right]$ to $\overline{\mathbf{Y_c}^n}$. This is done by
changing the channel matrix from the joint transmitters to the
cognitive receiver to $\mathbf{K} = \left[\begin{array}{cc}
\mathbf{0} & \mathbf{H_{c,c}}/\sqrt{\alpha} \end{array}\right]$.
\begin{lem}[\cite{Vishwanath2002}] The capacity region of the SMBCA is a superset of the capacity region of the scaled MIMO broadcast channel (SMBC).
\end{lem}
\begin{figure*}[!th]\label{fig : transformation 5}
\centering
\includegraphics[width=4.6in]{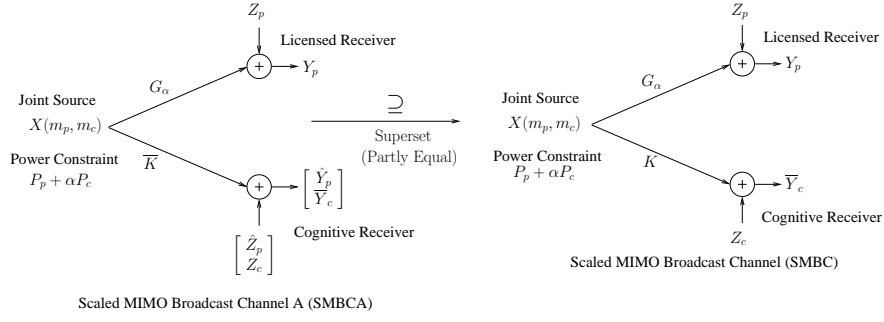}
\caption{Capacity Region of SMBCA $\supseteq$ Capacity Region of SMBC}
\end{figure*}
\begin{proof} :
Let the rate pair $(R_p, R_c)$ be achievable on the SMBC. That is, for all $\epsilon_1, \epsilon_2 > 0$, there exists
a $n$ and a sequence of encoder decoder pairs at the transmitter and the two receivers $(E^n : (m_p, m_c) \rightarrow
\mathbf{X}^n$, $D_p^n : \mathbf{Y_p}^n \rightarrow \hat{m}_p, D_c^n : \mathbf{Y_c}^n
\rightarrow \hat{m}_c)$ such that the codeword $\mathbf{X}^n$ satisfies the power constraint of $P_p + \alpha P_c$ and the probability of decoding error is small
$(Pr(m_p \neq \hat{m}_p) \leq \epsilon_1, Pr(m_c \neq \hat{m}_c)
\leq \epsilon_2)$.

In the SMBCA, the transmitter and the receivers use the same coding strategy. The licensed receiver can decode message $m_p$ at a rate $R_p$. The cognitive receiver can ignore $\mathbf{\hat{Y}_p}^n$ and use just $\mathbf{\overline{Y}_c}^n$ to decode message $m_c$ at a rate $R_c$. Hence, the rate pair $(R_p, R_c)$ is achievable in the SMBCA. Hence, the capacity region of the SMBCA is in general a superset of the capacity region of the SMBC.
\end{proof}

\vspace{0.5cm}
We describe one more lemma whose result will be used in the proof of Theorem (\ref{thm : partial converse}).
\vspace{0.5cm}
\begin{lem}[\cite{Vishwanath2002}] Let $\mathcal{C}_{SMBC}$ denote the capacity region of the scaled MIMO broadcast channel described in Figure $2$f. Then, for any $\mu \geq 1$,
\begin{displaymath}\sup_{(R_p, R_c) \in \mathcal{C}_{SMBC}} \mu R_p + R_c = \inf_{\mathbf{\Sigma_z}} \sup_{(R_p, R_c) \in \mathcal{C}_{SMBCA}} \mu R_p + R_c.\end{displaymath}
\end{lem}
\vspace{0.3cm}
The proof is described in \cite[Section 5.1]{Vishwanath2002} and is omitted here.

We now give the proof for Theorem (\ref{thm : partial converse}).

\begin{proof} \textit{of Theorem \ref{thm : partial converse}} : It was shown in \cite{Weingarten2006} that Gaussian codebooks (i.e., codebooks generated using i.i.d. realizations of an appropriate Gaussian random variable) achieve the capacity region for the MIMO broadcast channel.
In SMBC, let $\mathbf{Q_p}$ denote the covariance of codeword $\mathbf{X}^n$ for the licensed user and $\mathbf{Q_c}$ denote the covariance matrix for the cognitive user. The covariance matrices satisfy the joint power constraint $\mathrm{Tr}(\mathbf{Q_p} + \mathbf{Q_c}) \leq P_p + \alpha P_c$.
Let $\mathcal{R}_{SMBC,1}^{\alpha}$ denote the closure of the convex hull of the set of rate pairs described by
\begin{align}
\left\{\begin{array}{l}(R_p, R_c) : R_p \geq 0, R_c \geq 0 \vspace{0.15cm}\\
R_p \leq \log\left| \mathbf{I} + \mathbf{G_{\alpha}} \mathbf{Q_p} \mathbf{G_{\alpha}^{\dagger}} + \mathbf{G_{\alpha}} \mathbf{Q_c} \mathbf{G_{\alpha}^{\dagger}} \right|\vspace{0.05cm}\\
\qquad\quad - \log\left| \mathbf{I} + \mathbf{G_{\alpha}} \mathbf{Q_c} \mathbf{G_{\alpha}^{\dagger}}\right|\vspace{0.15cm}\\
R_c \leq \log\left|\mathbf{I} + \mathbf{K} \mathbf{Q_c} \mathbf{K^{\dagger}}\right|\vspace{0.15cm}\\
\forall \mathbf{Q_p} \succeq \mathbf{0}, \mathbf{Q_c} \succeq \mathbf{0}\vspace{0.15cm}\\
\mathrm{Tr}(\mathbf{Q_p}) + \mathrm{Tr}(\mathbf{Q_c}) \leq P_p + \alpha
P_c\end{array}\right\}.
\end{align}
Similarly, let $\mathcal{R}_{SMBC,2}^{\alpha}$ denote the closure of the convex hull of the set of rate pairs described by
\begin{align}
\left\{\begin{array}{l}(R_p, R_c) : R_p \geq 0, R_c \geq 0 \vspace{0.15cm}\\
R_p \leq \log\left|\mathbf{I} + \mathbf{G_{\alpha}} \mathbf{Q_p} \mathbf{G_{\alpha}^{\dagger}}\right|\vspace{0.15cm}\\
R_c \leq \log\left| \mathbf{I} + \mathbf{K} \mathbf{Q_p} \mathbf{K^{\dagger}} + \mathbf{K}\mathbf{Q_c}\mathbf{K^{\dagger}}\right|\vspace{0.05cm}\\
\qquad\quad - \log\left| \mathbf{I} + \mathbf{K} \mathbf{Q_p} \mathbf{K^{\dagger}}\right|\vspace{0.15cm}\\
\forall \mathbf{Q_p} \succeq \mathbf{0}, \mathbf{Q_c} \succeq \mathbf{0}, \vspace{0.15cm}\\
\mathrm{Tr}(\mathbf{Q_p}) + \mathrm{Tr}(\mathbf{Q_c}) \leq P_p + \alpha P_c \end{array}\right\}.
\end{align}
The capacity region of SMBC, $\mathcal{C}_{SMBC}$ is the closure of the convex hull of $\mathcal{R}_{SMBC,1}^{\alpha} \cup \mathcal{R}_{SMBC,2}^{\alpha}$. That is,
\begin{equation}
\mathcal{C}_{SMBC} = \mathrm{Cl}(\mathrm{Co}(\mathcal{R}_{SMBC,1}^{\alpha} \cup \mathcal{R}_{SMBC,2}^{\alpha})).
\end{equation}
$\mathcal{R}_{SMBC,1}^{\alpha}$ denotes the portion of the capacity region of SMBC where the licensed user's message is encoded first. That is, the cognitive receiver sees no interference. Hence, for $\mu \geq 1$, we have
\begin{displaymath}
\max_{(R_p, R_c) \in \mathcal{R}_{SMBC,1}^{\alpha}} \mu R_p + R_c = \max_{(R_p, R_c) \in \mathcal{C}_{SMBC}} \mu R_p + R_c.
\end{displaymath}
Therefore, from Lemma 5.6, we have that for $\mu \geq 1$,
\begin{displaymath}
\max_{(R_p, R_c) \in \mathcal{R}_{SMBC,1}^{\alpha}} \mu R_p + R_c = \inf_{\mathbf{\Sigma_z}}\max_{(R_p, R_c) \in \mathcal{C}_{SMBCA}} \mu R_p + R_c.
\end{displaymath}
We can see that, $\mathcal{R}_{part, out}^{\alpha}$ described in (\ref{eqn : partial converse region}) is a subset of $\mathcal{R}_{SMBC,1}^{\alpha}$ formed by restricting the covariance matrix $\mathbf{Q_c}$ to have the form
\begin{displaymath}
\mathbf{Q_c} = \left[\begin{array}{cc}\mathbf{0} & \mathbf{0} \\ \mathbf{0} & \mathbf{\Sigma_{c,c}}\end{array}\right].
\end{displaymath}
It can also be seen that $\mathcal{R}_{out}^{\alpha, \mathbf{\Sigma_z}}$ described in (\ref{eqn : converse region}) equals $\mathcal{C}_{SMBCA}$. Hence, it follows that for any $\mu \geq 1$ and for $\alpha > 0$, if
\begin{displaymath}
\max_{(R_p, R_c) \in \mathcal{R}_{part, out}^{\alpha}} \mu R_p + R_c = \max_{(R_p, R_c) \in \mathcal{C}_{BC}^{\alpha}} \mu R_p + R_c,
\end{displaymath}
then we have that
\begin{displaymath}
\max_{(R_p, R_c) \in \mathcal{R}_{part, out}^{\alpha}} \mu R_p + R_c = \inf_{\mathbf{\Sigma_z}}\max_{(R_p, R_c) \in \mathcal{R}_{out}^{\alpha, \mathbf{\Sigma_z}}} \mu R_p + R_c.
\end{displaymath}
\end{proof}
\section{Optimality of the Achievable Region} \label{sec : optimality}
In this section, we describe conditions under which the achievable region described by
$\mathcal{R}_{in}$ in (\ref{eqn : achievable region}) is optimal for a
portion of the capacity region. In particular, we show that if $(R_p, R_c)$ lies on the boundary of the achievable region given by
$\mathcal{R}_{in}$, then $(R_p, R_c)$ lies on the boundary of $\mathcal{R}_{part, out}^{\alpha}$ given by (\ref{eqn : partial converse region}) for some $\alpha > 0$. That is, for any $\mu > 0$,
\begin{displaymath}
\sup_{(R_p, R_c) \in \mathcal{R}_{in}} \mu R_p + R_c = \inf_{\alpha > 0} \sup_{(R_p, R_c) \in \mathcal{R}_{part, out}^{\alpha}} \mu R_p + R_c.
\end{displaymath}
Then there exists $\alpha^* \in (0, \infty)$ such that, for any $\mu \geq 1$, $(R_{p,\mu}, R_{c,\mu}) = \arg \max_{(R_p, R_c) \in
\mathcal{R}_{in}} \mu R_p + R_c$ is a point on the boundary of the
capacity region of the MIMO cognitive channel if the condition (\ref{eqn : condition}) is satisfied for $\alpha^*$.

We denote by $\mathcal{R}_{ach,rate}$, the set of all
$((R_p, R_c), \mathbf{\Sigma_p}, \mathbf{\Sigma_{c,p}}, \mathbf{\Sigma_{c,c}}, \mathbf{Q})$ given by
\begin{align} \label{eqn : describe achievable region rate}
\left\{\begin{array}{l}\bigg((R_p, R_c), \mathbf{\Sigma_p}, \mathbf{\Sigma_{c,p}}, \mathbf{\Sigma_{c,c}}, \mathbf{Q}\bigg) : \vspace{0.15cm}\\
R_p \geq 0, R_c \geq 0,\mathbf{\Sigma_p} \succeq \mathbf{0}, \mathbf{\Sigma_{c,p}} \succeq \mathbf{0}, \mathbf{\Sigma_{c,c}} \succeq \mathbf{0}\vspace{0.15cm}\\
R_p \leq \log\left| \mathbf{I} + \mathbf{G} \mathbf{\Sigma_{p,net}} \mathbf{G^{\dagger}} + \mathbf{H_{c,p}} \mathbf{\Sigma_{c,c}} \mathbf{H_{c,p}^{\dagger}}\right|\vspace{0.05cm}\\
\qquad\quad  - \log\left| \mathbf{I} + \mathbf{H_{c,p}\Sigma_{c,c} H_{c,p}^{\dagger}}\right|\vspace{0.15cm}\\
R_c \leq \log\left| \mathbf{I} + \mathbf{H_{c,c} \Sigma_{c,c} H_{c,c}^{\dagger}}\right| \vspace{0.15cm}\\
\mathbf{\Sigma_{p,net}} = \left(\begin{array}{ll} \mathbf{\Sigma_p} & \mathbf{Q} \\ \mathbf{Q^{\dagger}}  & \mathbf{\Sigma_{c,p}} \end{array}\right) \succeq \mathbf{0} \end{array}\right\}.
\end{align}
The rate pair that maximizes $\mu R_p + R_c$ in the achievable region is given by solving the optimization problem
\begin{align} \label{eqn : optimization 1}
\sup_{((R_p, R_c), \mathbf{\Sigma_p}, \mathbf{\Sigma_{c,p}}, \mathbf{\Sigma_{c,c}}, \mathbf{Q})}& \mu R_p + R_c\\
\textrm{such that}\quad\qquad & ((R_p, R_c), \mathbf{\Sigma_p}, \mathbf{\Sigma_{c,p}}, \mathbf{\Sigma_{c,c}}, \mathbf{Q})\nonumber\\
& \qquad\qquad\in \mathcal{R}_{ach,rate}\nonumber\\
& \mathrm{Tr}(\mathbf{\Sigma_{p}}) \leq P_p\nonumber\\
& \mathrm{Tr}(\mathbf{\Sigma_{c,p}} + \mathbf{\Sigma_{c,c}}) \leq P_c. \nonumber
\end{align}
We define the functions $L(R_p, R_c, \mathbf{\Sigma_p}, \mathbf{\Sigma_{c,p}}, \mathbf{\Sigma_{c,c}},  \lambda_1, \lambda_2)$ and $g(R_p, R_c, \mathbf{\Sigma_p}, \mathbf{\Sigma_{c,p}}, \mathbf{\Sigma_{c,c}})$ as follows
\begin{align}\label{eqn : L}
L(R_p, R_c, \mathbf{\Sigma_p}, \mathbf{\Sigma_{c,p}}, \mathbf{\Sigma_{c,c}},  \lambda_1, \lambda_2) = \qquad\qquad\qquad\nonumber\\
\qquad\mu R_p + R_c -\lambda_1(\mathrm{Tr}(\mathbf{\Sigma_p}) - P_p) -\\
\qquad\quad\lambda_2 (\mathrm{Tr}(\mathbf{\Sigma_{c,p}} + \mathbf{\Sigma_{c,c}}) - P_c))\nonumber
\end{align}
\begin{align}\label{eqn : g}
g(R_p, R_c, \mathbf{\Sigma_p}, \mathbf{\Sigma_{c,p}}, \mathbf{\Sigma_{c,c}}) = \qquad\qquad\qquad\qquad\qquad\nonumber\\
\qquad\qquad \min _{\lambda_1 \geq 0, \lambda_2 \geq 0}  L(R_p, R_c, \mathbf{\Sigma_p}, \mathbf{\Sigma_{c,p}}, \mathbf{\Sigma_{c,c}}, \lambda_1, \lambda_2).
\end{align}
The optimization problem given by
\begin{align}\label{eqn : optimization 2}
\max_{(R_p, R_c, \mathbf{\Sigma_p}, \mathbf{\Sigma_{c,p}}, \mathbf{\Sigma_{c,c}}, \mathbf{Q})} &\  g(R_p, R_c, \mathbf{\Sigma_p}, \mathbf{\Sigma_{c,p}}, \mathbf{\Sigma_{c,c}})\\
\textrm{such that}\quad\qquad & ((R_p, R_c), \mathbf{\Sigma_p}, \mathbf{\Sigma_{c,p}}, \mathbf{\Sigma_{c,c}}, \mathbf{Q})\nonumber\\
& \qquad\qquad \in \mathcal{R}_{ach,rate}\nonumber
\end{align}
has the same optimum value as that of (\ref{eqn : optimization 1}). This is formally stated in the lemma below.
\begin{lem}\label{lem : M = U}
Let $M$ denote the optimal value of the optimization problem defined
in (\ref{eqn : optimization 1}), and $U$ denote the optimal value of
the optimization problem defined in (\ref{eqn : optimization 2}).
Then, $M = U$.
\end{lem}
\begin{proof} :
We show that for any set of covariance matrices $(\mathbf{\Sigma_p}, \mathbf{\Sigma_{c,p}}, \mathbf{\Sigma_{c,c}})$ that do not satisfy the power
constraints given by (\ref{eqn : power constraint}), $g(R_p, R_c, \mathbf{\Sigma_p}, \mathbf{\Sigma_{c,p}}, \mathbf{\Sigma_{c,c}}) = -\infty$. The power constraints can be violated by three means :
\begin{itemize}
\item $\mathrm{Tr}(\mathbf{\Sigma_p}) > P_p$ and $\mathrm{Tr}(\mathbf{\Sigma_{c,p}}) + \mathrm{Tr}(\mathbf{\Sigma_{c,c}}) \leq P_c$ : In this case, $\lambda_1$ will take an arbitrarily large value and $\lambda_2 = 0$ to drive $g(R_p, R_c, \mathbf{\Sigma_p}, \mathbf{\Sigma_{c,p}}, \mathbf{\Sigma_{c,c}})$ to $-\infty$.
\item $\mathrm{Tr}(\mathbf{\Sigma_p}) \leq P_p$ and $\mathrm{Tr}(\mathbf{\Sigma_{c,p}}) + \mathrm{Tr}(\mathbf{\Sigma_{c,c}}) > P_c$ : In this case, $\lambda_1 = 0$ and $\lambda_2$ will take an arbitrarily large value to drive $g(R_p, R_c, \mathbf{\Sigma_p}, \mathbf{\Sigma_{c,p}}, \mathbf{\Sigma_{c,c}})$ to $-\infty$.
\item $\mathrm{Tr}(\mathbf{\Sigma_p}) > P_p$ and $\mathrm{Tr}(\mathbf{\Sigma_{c,p}}) + \mathrm{Tr}(\mathbf{\Sigma_{c,c}}) > P_c$ : In this case, $\lambda_1$ and $\lambda_2$ will take arbitrarily large values to drive $g(R_p, R_c, \mathbf{\Sigma_p}, \mathbf{\Sigma_{c,p}}, \mathbf{\Sigma_{c,c}})$ to $-\infty$.
\end{itemize}
When both the covariance matrices satisfy the power constraints with inequality, then $\lambda_1 = \lambda_2 = 0$. This is because, $\mathrm{Tr}(\mathbf{\Sigma_p}) - P_p$ and $\mathrm{Tr}(\mathbf{\Sigma_{c,p}} + \mathbf{\Sigma_{c,c}}) - P_c$ are both negative. Hence, for any positive value of $\lambda_1$ or $\lambda_2$, $L(R_p, R_c, \mathbf{\Sigma_p}, \mathbf{\Sigma_{c,p}}, \mathbf{\Sigma_{c,c}},  \lambda_1, \lambda_2) \geq L(R_p, R_c, \mathbf{\Sigma_p}, \mathbf{\Sigma_{c,p}}, \mathbf{\Sigma_{c,c}}, 0, 0)$.

When one of the power constraint is satisfied with equality, say $\mathrm{Tr}(\mathbf{\Sigma_p}) - P_p = 0$ and the other power constraint is satisfied with inequality $\mathrm{Tr}(\mathbf{\Sigma_{c,p}} + \mathbf{\Sigma_{c,c}}) - P_c < 0$, then, we will have $\lambda_2 = 0$ and $\lambda_1$ will be some real number. In any case, we still have $\lambda_1 (\mathrm{Tr}(\mathbf{\Sigma_p}) - P_p) = \lambda_2 (\mathrm{Tr}(\mathbf{\Sigma_{c,p}} + \mathbf{\Sigma_{c,c}}) - P_c) = 0$.

Similarly, when the first constraint is satisfied with inequality, and the second constraint satisfied with equality, we have $\lambda_1 = 0$ and $\lambda_2$ is some non negative real number. We have $\lambda_1 (\mathrm{Tr}(\mathbf{\Sigma_p}) - P_p) = \lambda_2 (\mathrm{Tr}(\mathbf{\Sigma_{c,p}} + \mathbf{\Sigma_{c,c}}) - P_c) = 0$.

Finally, if both the power constraints are satisfied with equality, $\lambda_1$ and $\lambda_2$ are some non-negative real numbers. And $\lambda_1 (\mathrm{Tr}(\mathbf{\Sigma_p}) - P_p) = \lambda_2 (\mathrm{Tr}(\mathbf{\Sigma_{c,p}} + \mathbf{\Sigma_{c,c}}) - P_c) = 0$.

Hence, in all the cases, the complementary slackness conditions are satisfied. Hence, the optimal solution of the optimization problem
(\ref{eqn : optimization 2}) satisfy the power constraints and the
objective function reduces to that of optimization problem (\ref{eqn
: optimization 1}). Hence, both the optimization problems have the
same optimal values. That is, $M = U$.
\end{proof}

Next, we find the optimum value of $\mu R_p + R_c$ over all the rate pairs that are in the region $\mathcal{R}_{part,out}^{\alpha}$ described by (\ref{eqn : partial converse region}). This is done by solving the following optimization problem:
\begin{align} \label{eqn : optimization 3}
\sup_{((R_p, R_c), \mathbf{Q_p}, \mathbf{\Sigma_{c,c}})} & \ \  \mu R_p + R_c \\
\textrm{such that}\qquad & ((R_p, R_c), \mathbf{Q_p}, \mathbf{\Sigma_{c,c}}) \in \mathcal{R}_{part, conv, rate}^{\alpha}\nonumber\\
&\mathrm{Tr}(\mathbf{\Sigma_{c,c}}) + \mathrm{Tr}(\mathbf{Q_p}) \leq \alpha P_c + P_p,\nonumber
\end{align}
where $\mathcal{R}_{part, conv, rate}^{\alpha}$ is the set of quadruples $((R_p, R_c), \mathbf{Q_p}, \mathbf{\Sigma_{c,c}})$ described by
\begin{align}\label{eqn : Rpart,conv,rate}
\left\{\begin{array}{l} ((R_p, R_c), \mathbf{Q_p}, \mathbf{\Sigma_{c,c}}) : \vspace{0.15cm}\\
R_p \geq 0, R_c \geq 0, \mathbf{Q_p} \succeq \mathbf{0}, \mathbf{\Sigma_{c,c}} \succeq \mathbf{0}\vspace{0.15cm}\\
R_p \leq \log\left| \mathbf{I} + \mathbf{G_{\alpha}} \mathbf{Q_p} \mathbf{G_{\alpha}^{\dagger}} + \frac{1}{\alpha}\mathbf{H_{c,p}} \mathbf{\Sigma_{c,c}} \mathbf{H_{c,p}^{\dagger}}\right|\vspace{0.05cm}\\
\qquad\quad - \log\left| \mathbf{I} + \frac{1}{\alpha}\mathbf{H_{c,p}\Sigma_{c,c} H_{c,p}^{\dagger}}\right|\vspace{0.15cm}\\
R_c \leq \log\left| \mathbf{I} + \frac{1}{\alpha}\mathbf{H_{c,c} \Sigma_{c,c} H_{c,c}^{\dagger}}\right| \end{array}\right\}.
\end{align}
We let the optimal solution of (\ref{eqn : optimization 3}) to be
denoted by $N(\alpha)$.
Let $N = \min_{\alpha > 0} N(\alpha)$ and
\begin{equation}\label{eqn : alpha^*}
\alpha^* = \arg\min_{\alpha > 0} N(\alpha).
\end{equation}
We show in Lemma 6.2 that $\alpha^* \in (0, \infty)$ exists. Then, $N$ is given by the optimum value of the following $\inf \sup$ optimization problem
\begin{align} \label{eqn : optimization 4}
\inf_{\alpha > 0}\ \ \ & \qquad\sup_{((R_p, R_c), \mathbf{Q_p}, \mathbf{\Sigma_{c,c}})}\ \  \mu R_p + R_c \\
\textrm{such that}& \qquad ((R_p, R_c), \mathbf{Q_p}, \mathbf{\Sigma_{c,c}}) \in \mathcal{R}_{part, conv, rate}^{\alpha}\nonumber\\
&\qquad\mathrm{Tr}(\mathbf{\Sigma_{c,c}}) + \mathrm{Tr}(\mathbf{Q_p}) \leq \alpha P_c + P_p.\nonumber
\end{align}
The infimum constraint $\alpha > 0$ is not a compact set. We modify the constraint on $\alpha$ to $\alpha \in
\mathbb{R}^+ \cup \{0, \infty\}$. This is done to compactify the set
by adding two extra symbols $0$ and $\infty$. The point zero is added to make the set closed. The process of adding the point $\infty$ is called one point compactification. Details on one point compactification can be found in \cite[Section 2.8]{DudleyRealAnalysisProbability}. The new space $\alpha \in \mathbb{R}^+ \cup \{0, \infty\}$ is compact and Hausdorff.

The optimization problem after changing the constraint set on $\alpha$ becomes
\begin{align} \label{eqn : optimization 4a}
N_1 = \inf_{\alpha \in \mathbb{R}^+ \cup \{0, \infty\}} & \quad\sup_{((R_p, R_c), \mathbf{Q_p}, \mathbf{\Sigma_{c,c}})}\quad \mu R_p + R_c \\
\textrm{such that}\quad & \quad((R_p, R_c), \mathbf{Q_p}, \mathbf{\Sigma_{c,c}}) \in \mathcal{R}_{part, conv, rate}^{\alpha}\nonumber\\
&\quad\mathrm{Tr}(\mathbf{\Sigma_{c,c}}) + \mathrm{Tr}(\mathbf{Q_p}) \leq \alpha P_c + P_p.\nonumber
\end{align}
We show that adding the two points $0$ and $\infty$ to the constraint set on $\alpha$ does not change the optimum value of the optimization problem. This result is formally stated and proved in the following lemma.

\begin{lem} \label{lem : compactification}
The optimum value of the optimization problem given by (\ref{eqn : optimization 4}), $N$ is equal to the optimum value of the optimization problem described by (\ref{eqn : optimization 4a}), $N_1$. That is, $N = N_1$.
\end{lem}
\begin{proof} : For any $\alpha \in \mathbb{R}^+ \cup \{0, \infty\}$, we let $h(\alpha)$ to denote the value of the inner $\sup$ problem. That is,
\begin{align}\label{eqn : h(alpha)}
h(\alpha) =& \sup_{((R_p, R_c), \mathbf{Q_p}, \mathbf{\Sigma_{c,c}})} \quad\mu R_p + R_c \\
\textrm{such that }&\quad ((R_p, R_c), \mathbf{Q_p}, \mathbf{\Sigma_{c,c}}) \in \mathcal{R}_{part, conv, rate}^{\alpha}\nonumber\\
&\quad\mathrm{Tr}(\mathbf{\Sigma_{c,c}}) + \mathrm{Tr}(\mathbf{Q_p}) \leq P_p + \alpha P_c.\nonumber
\end{align}
We show that $\liminf_{\alpha \rightarrow 0} h(\alpha) = \liminf_{\alpha \rightarrow \infty} h(\alpha) = \infty$.

Letting $\alpha \rightarrow 0$, we put all the power in $\mathbf{\Sigma_{c,c}}$. That is, we choose $\mathbf{\Sigma_{p}} = \mathbf{0}$, $\mathbf{\Sigma_{c,p}} = \mathbf{0}$, $\mathbf{Q} = \mathbf{0}$ and $\mathbf{\Sigma_{c,c}} = \frac{P_p + \alpha P_c}{n_{c,t}} \mathbf{I}_{n_{c,t}}$. Also, we take
\begin{displaymath}
R_p = 0 \textrm{ and } R_c = \log\left| \mathbf{I} + \frac{1}{\alpha}\frac{P_p + \alpha P_c}{n_{c,t}} \mathbf{H_{c,c}}\mathbf{H_{c,c}^{\dagger}}\right|.
\end{displaymath}
It follows from (\ref{eqn : Rpart,conv,rate}) that $((R_p, R_c), \mathbf{Q_p}, \mathbf{\Sigma_{c,c}}) \in \mathcal{R}_{part, conv, rate}^{\alpha}$. Also, $\mathrm{Tr}(\mathbf{Q_p}) + \mathrm{Tr}(\mathbf{\Sigma_{c,c}}) = P_p + \alpha P_c$. Hence, $((R_p, R_c), \mathbf{Q_p}, \mathbf{\Sigma_{c,c}}) $ satisfy all the necessary constraints of (\ref{eqn : h(alpha)}). Substituting these particular values of $((R_p, R_c), \mathbf{Q_p}, \mathbf{\Sigma_{c,c}})$, we get a lower bound on $h(\alpha)$. That is,
\begin{align}
\liminf_{\alpha \rightarrow 0} h(\alpha)  \geq &\liminf_{\alpha \rightarrow 0} \log\left| \mathbf{I} + \frac{1}{\alpha}\frac{P_p + \alpha P_c}{n_{c,t}} \mathbf{H_{c,c}}\mathbf{H_{c,c}^{\dagger}}\right|\nonumber\\
= &\qquad\infty.
\end{align}
Next, we look at the situation when $\alpha \rightarrow \infty$. In this case, we put all the power in $\mathbf{\Sigma_p}$. That is, we choose $\mathbf{\Sigma_p} = \frac{P_p + \alpha P_c}{n_{p,t}} \mathbf{I}_{n_{p,t}}$, $\mathbf{\Sigma_{c,p}} =  \mathbf{0}$, $\mathbf{\Sigma_{c,c}} = \mathbf{0}$ and $\mathbf{Q} = \mathbf{0}$. We also choose
\begin{displaymath}
R_c = 0 \textrm{ and } R_p = \log\left| \mathbf{I} + \frac{P_p + \alpha P_c}{n_{p,t}} \mathbf{H_{p,p}}\mathbf{H_{p,p}^{\dagger}}\right|.
\end{displaymath}
These values of $((R_p, R_c), \mathbf{Q_p}, \mathbf{\Sigma_{c,c}})$ satisfy all the necessary constraints of (\ref{eqn : h(alpha)}). Hence, we have
\begin{align}
\liminf_{\alpha \rightarrow \infty} h(\alpha) \geq &\liminf_{\alpha \rightarrow \infty} \mu \log\left| \mathbf{I} + \frac{P_p + \alpha P_c}{n_{p,t}} \mathbf{H_{p,p}}\mathbf{H_{p,p}^{\dagger}}\right|\nonumber\\
=& \qquad\infty.
\end{align}
Hence, $h(\alpha) = \infty$ when $\alpha = 0$ or $\alpha = \infty$. Also, when $\alpha \in \mathbb{R}^+$, $h(\alpha) < \infty$. Hence, the optimum value of (\ref{eqn : optimization 4a}) is reached when $\alpha$ is neither $0 \textrm{ nor } \infty$. Hence, $N = N_1$.
\end{proof}
As $\mathbf{Q_p}$ is the covariance matrix of the codeword
$\mathbf{X}(i), i = 1, \ldots, n$ for the primary user, it can be
written as
\begin{equation}
\mathbf{Q_p} = \left(\begin{array}{cc} \mathbf{\Sigma_p} &
\mathbf{Q} \\ \mathbf{Q^{\dagger}} &
\mathbf{\Sigma_{c,p}}\end{array}\right).
\end{equation}
It is easy to see that the set $R_{part, conv}^{\alpha}$ described in (11) can also be written as
\vspace{0.1cm}
\begin{eqnarray}\label{eqn : describe converse region 1}
\left\{\begin{array}{l}\bigg((R_p, R_c), \mathbf{\Sigma_p}, \mathbf{\Sigma_{c,p}}
, \mathbf{Q}, \mathbf{\Sigma_{c,c}}\bigg) : \vspace{0.15cm}\\
R_p \geq 0, R_c \geq 0, \mathbf{\Sigma_p}\succeq \mathbf{0}, \mathbf{\Sigma_{c,p}} \succeq \mathbf{0}, \mathbf{\Sigma_{c,c}} \succeq \mathbf{0}\vspace{0.15cm}\\
R_p \leq \log\left| \mathbf{I} + \mathbf{G} \mathbf{Q_p} \mathbf{G^{\dagger}} + \mathbf{H_{c,p}} \mathbf{\Sigma_{c,c}}\mathbf{H_{c,p}^{\dagger}}\right|\vspace{0.05cm}\\
\qquad\quad - \log\left| \mathbf{I} + \mathbf{H_{c,p}\Sigma_{c,c} H_{c,p}^{\dagger}}\right|\vspace{0.15cm}\\
R_c \leq \log\left| \mathbf{I} + \mathbf{H_{c,c} \Sigma_{c,c} H_{c,c}^{\dagger}}\right| \vspace{0.15cm}\\
\mathrm{Tr}(\mathbf{\Sigma_p}) + \alpha \mathrm{Tr}(\mathbf{\Sigma_{c,p}}) + \alpha \mathrm{Tr}(\mathbf{\Sigma_{c,c}}) \leq P_p + \alpha
P_c\end{array}\right\}.
\end{eqnarray}
where $\mathbf{G} = \left[\mathbf{H_{p,p}}\ \ \mathbf{H_{c,p}}\right]$. This is done by transforming $\mathbf{Q}, \mathbf{\Sigma_{c,p}}, \mathbf{\Sigma_{c,c}}$ into $\sqrt{\alpha} \mathbf{Q}, \alpha \mathbf{\Sigma_{c,p}}, \alpha \mathbf{\Sigma_{c,c}}$ respectively. We define $\mathcal{R}_{part, conv, rate}$ as the set described by
\vspace{0.1cm}
\begin{align}
\left\{\begin{array}{l} ((R_p, R_c), \mathbf{\Sigma_p}, \mathbf{\Sigma_{c,p}}, \mathbf{Q}, \mathbf{\Sigma_{c,c}}) : \vspace{0.15cm}\\
R_p \geq 0, R_c \geq 0, \mathbf{\Sigma_p} \succeq \mathbf{0}, \mathbf{\Sigma_{c,p}} \succeq \mathbf{0}, \mathbf{\Sigma_{c,c}} \succeq \mathbf{0}\vspace{0.15cm}\\
R_p \leq \log\left| \mathbf{I} + \mathbf{G} \mathbf{Q_p} \mathbf{G^{\dagger}} + \mathbf{H_{c,p}} \mathbf{\Sigma_{c,c}} \mathbf{H_{c,p}^{\dagger}}\right| \vspace{0.05cm}\\
\qquad\quad - \log\left| \mathbf{I} + \mathbf{H_{c,p}\Sigma_{c,c} H_{c,p}^{\dagger}}\right|\vspace{0.15cm}\\
R_c \leq \log\left| \mathbf{I} + \mathbf{H_{c,c} \Sigma_{c,c} H_{c,c}^{\dagger}}\right|, \vspace{0.15cm}\\
\mathbf{Q_p} = \left(\begin{array}{cc} \mathbf{\Sigma_p} &
\mathbf{Q} \\ \mathbf{Q^{\dagger}} &
\mathbf{\Sigma_{c,p}}\end{array}\right)\end{array}\right\}.
\end{align}
Hence, the optimization problem (\ref{eqn : optimization 4a}) can be written as
\begin{align} \label{eqn : optimization 5}
N = &\inf_{\alpha \in \mathbb{R}^+ \cup \{0, \infty\}}\ \sup_{((R_p, R_c), \mathbf{\Sigma_p}, \mathbf{\Sigma_{c,p}}, \mathbf{Q}, \mathbf{\Sigma_{c,c}})}\ \mu R_p + R_c\\
\textrm{such that}&\quad ((R_p, R_c), \mathbf{\Sigma_p}, \mathbf{\Sigma_{c,p}}, \mathbf{Q}, \mathbf{\Sigma_{c,c}}) \in \mathcal{R}_{part,conv,rate}\nonumber\\
&\quad\mathrm{Tr}(\mathbf{\Sigma_p}) + \alpha \mathrm{Tr}(\mathbf{\Sigma_{c,p}}) + \alpha \mathrm{Tr}(\mathbf{\Sigma_{c,c}}) \leq P_p + \alpha P_c.\nonumber
\end{align}
We state the following lemma for switching $\min$ and $\max$ in minimax problems. The lemma is described and proved in Theorem 2 in \cite{Ky-Fan}.
\begin{lem} (Ky-Fan's minimax switching theorem \cite[Thm. 2]{Ky-Fan})
Let $X$ be a compact Hausdorff space and $Y$ an arbitrary set (not topologized). Let $f$ be a real-valued function on $X \times Y$ such that, for every $y \in Y$, $f(x, y)$ is lower semi continuous on $X$. If $f$ is convex on $X$ and concave on $Y$, then
\begin{equation}
\inf_{x \in X}\  \sup_{y \in Y} f(x, y) = \sup_{y \in Y}\  \inf_{x \in X} f(x, y).\footnote{In (49), the $\inf$ can be replaced with $\min$, but we use $\inf$ throughout to maintain continuity and to avoid confusion.}
\end{equation}
\end{lem}
We see that the objective function $\mu R_p + R_c$ is concave with respect to the maximizing
variables $((R_p, R_c, \mathbf{Q_p}, \mathbf{\Sigma_{c,c}})$ and convex with respect to the minimizing variable $\alpha$. The constraint space $\alpha \in \mathbb{R}^+ \cup \{0, \infty\}$ is compact and Hausdorff \cite[Section 2.8]{DudleyRealAnalysisProbability}.

Hence, all the conditions of the lemma are satisfied. Hence, by Ky-Fan's mini-max switching theorem \cite{Ky-Fan}, we can interchange the $\sup$ and $\inf$ without affecting the optimum value. Hence,
\begin{align} \label{eqn : optimization 6}
N = &\sup_{((R_p, R_c), \mathbf{\Sigma_p}, \mathbf{\Sigma_{c,p}},
\mathbf{Q}, \mathbf{\Sigma_{c,c}})}
\ \inf_{\alpha \in \mathbb{R}^+ \cup \{0, \infty\}}\ \mu R_p + R_c \\
\textrm{such that} &\  ((R_p, R_c), \mathbf{\Sigma_p}, \mathbf{\Sigma_{c,p}},
\mathbf{Q}, \mathbf{\Sigma_{c,c}})\in \mathcal{R}_{part, conv,rate}\nonumber\\
& \mathrm{Tr}(\mathbf{\Sigma_p}) + \alpha \mathrm{Tr}(\mathbf{\Sigma_{c,p}})
+ \alpha \mathrm{Tr}(\mathbf{\Sigma_{c,c}}) \leq P_p + \alpha P_c.\nonumber
\end{align}
Similar to the functions $L$ and $g$ defined in (\ref{eqn : L}) and
(\ref{eqn : g}), we define the functions $L_1(R_p, R_c,
\mathbf{\Sigma_p}, \mathbf{\Sigma_{c,p}}, \mathbf{\Sigma_{c,c}},
\lambda,\alpha)$ and $g_1(R_p, R_c, \mathbf{\Sigma_p},
\mathbf{\Sigma_{c,p}}, \mathbf{\Sigma_{c,c}}, \alpha)$ as follows
\begin{align}\label{eqn : L1}
L_1(R_p, R_c, \mathbf{\Sigma_p}, \mathbf{\Sigma_{c,p}}, \mathbf{\Sigma_{c,c}}, \lambda,\alpha) = \mu R_p + R_c - \qquad\ \nonumber\\
\lambda\bigg(\mathrm{Tr} (\mathbf{\Sigma_p}) + \alpha \mathrm{Tr}(\mathbf{\Sigma_{c,p}}) + \alpha \mathrm{Tr}(\mathbf{\Sigma_{c,c}}) - P_p - \alpha P_c\bigg),
\end{align}
\begin{align}\label{eqn : g1}
g_1(R_p, R_c, \mathbf{\Sigma_p}, \mathbf{\Sigma_{c,p}}, \mathbf{\Sigma_{c,c}},\alpha) = \qquad\qquad\qquad\qquad\nonumber\\
 \qquad\qquad\inf _{\lambda \geq 0}  L_1(R_p, R_c, \mathbf{\Sigma_p}, \mathbf{\Sigma_{c,p}}, \mathbf{\Sigma_{c,c}},  \lambda, \alpha).
\end{align}
We define the following optimization problem
\begin{align}\label{eqn : optimization 7}
V = \sup_{(R_p, R_c, \mathbf{\Sigma_p}, \mathbf{\Sigma_{c,p}},
\mathbf{Q}, \mathbf{\Sigma_{c,c}})}
\inf_{\alpha}  g_1(R_p, R_c, \mathbf{\Sigma_p}, \mathbf{\Sigma_{c,p}},
\mathbf{\Sigma_{c,c}},\alpha)\\
\textrm{such that}\quad  ((R_p, R_c), \mathbf{\Sigma_p}, \mathbf{\Sigma_{c,p}},
\mathbf{Q}, \mathbf{\Sigma_{c,c}}) \in \mathcal{R}_{part,conv,rate}\qquad\nonumber\\
\quad \alpha \in \mathbb{R}^+ \cup \{0, \infty\}.\qquad\qquad\qquad\qquad\qquad\qquad\nonumber
\end{align}
\begin{lem}\label{lem : N = V}
The optimum value of optimization problem (\ref{eqn : optimization
6}), N is equal to the optimum value of the optimization problem
(\ref{eqn : optimization 7}), V.
\end{lem}
\begin{proof} :
The proof of the lemma is along the same lines as the proof of Lemma \ref{lem : M = U}. We show that for any set of covariance matrices $\mathbf{\Sigma_p}$, $\mathbf{\Sigma_{c,p}}$ and $\mathbf{\Sigma_{c,c}}$ that do not satisfy the power constraint $\mathrm{Tr}(\mathbf{\Sigma_p}) + \alpha \mathrm{Tr}(\mathbf{\Sigma_{c,p}}) + \alpha \mathrm{Tr}(\mathbf{\Sigma_{c,c}}) \leq P_p + \alpha P_c$, $g_1(R_p, R_c, \mathbf{\Sigma_p}, \mathbf{\Sigma_{c,p}}, \mathbf{\Sigma_{c,c}}, \alpha) = -\infty$. This is because, $\mathrm{Tr}(\mathbf{\Sigma_p}) + \alpha \mathrm{Tr}(\mathbf{\Sigma_{c,p}}) + \alpha \mathrm{Tr}(\mathbf{\Sigma_{c,c}}) - P_p -\alpha P_c$ is positive, and hence, $\lambda$ will take an arbitrarily high value to drive $g_1(R_p, R_c, \mathbf{\Sigma_p}, \mathbf{\Sigma_{c,p}}, \mathbf{\Sigma_{c,c}}, \alpha)$ to $-\infty$. Hence, the outer supremization problem will ensure that the power constraint is satisfied.

Moreover, when the power constraints are satisfied with inequality, then $\mathrm{Tr}(\mathbf{\Sigma_p}) + \alpha \mathrm{Tr}(\mathbf{\Sigma_{c,p}}) + \alpha \mathrm{Tr}(\mathbf{\Sigma_{c,c}}) - P_p -\alpha P_c$ is negative. Therefore, for any $\lambda > 0$,
we have $L_1(R_p, R_c, \mathbf{\Sigma_p}, \mathbf{\Sigma_{c,p}}, \mathbf{\Sigma_{c,c}}, \lambda,\alpha)  > L_1(R_p, R_c, \mathbf{\Sigma_p}, \mathbf{\Sigma_{c,p}}, \mathbf{\Sigma_{c,c}}, 0, \alpha) $. Hence, $\lambda$ will take the value zero.
When the power constraint is satisfied with equality, then $\mathrm{Tr}(\mathbf{\Sigma_p}) + \alpha \mathrm{Tr}(\mathbf{\Sigma_{c,p}}) + \alpha \mathrm{Tr}(\mathbf{\Sigma_{c,c}}) - P_p -\alpha P_c = 0$. Then, $\lambda$ will take some non negative real number. Hence, the complementary slackness condition is satisfied.
Hence, the optimal solution of the optimization problem satisfy the power constraint and the objective function reduces to that of (\ref{eqn : optimization 6}). It follows that, the optimum value of the optimization problem (\ref{eqn : optimization 6}), $N$ is the same as the optimum value of the optimization problem (\ref{eqn : optimization 7}), $V$.
\end{proof}
Next, we show that the optimum value of the optimization problem
(\ref{eqn : optimization 2}), $U$ is an upper bound on the optimal
value of the optimization problem (\ref{eqn : optimization 7}), $V$.
\begin{lem}\label{lem : 6.5} The optimal value of (\ref{eqn : optimization 2}), $U$ is an upper bound on the optimal value of (\ref{eqn : optimization 6}), $V$. \end{lem}
\begin{proof} :
Both the optimization problems are $\sup \min$ problems. For any
$\lambda_1 \geq 0$ and $\lambda_2 \geq 0$, we can choose $\lambda =
\lambda_1$ and $\alpha = \lambda_2/ \lambda_1$ so that $L_1(R_p,
R_c, \mathbf{\Sigma_p}, \mathbf{\Sigma_{c,p}},
\mathbf{\Sigma_{c,c}}, \lambda, \alpha) = L(R_p, R_c,
\mathbf{\Sigma_p}, \mathbf{\Sigma_{c,p}}, \mathbf{\Sigma_{c,c}},
\lambda_1, \lambda_2)$. Hence, for any $((R_p, R_c),
\mathbf{\Sigma_p}, \mathbf{\Sigma_{c,p}}, \mathbf{\Sigma_{c,c}})$,
\begin{align}
\inf_{\lambda \geq 0, \alpha \in \mathbb{R}^+ \cup \{0, \infty\}}
L_1(R_p, R_c, \mathbf{\Sigma_p}, \mathbf{\Sigma_{c,p}},
\mathbf{\Sigma_{c,c}}, \lambda, \alpha) \leq\nonumber\\
\qquad \inf_{\lambda_1 \geq 0,
\lambda_2 \geq 0} L(R_p, R_c, \mathbf{\Sigma_p},
\mathbf{\Sigma_{c,p}}, \mathbf{\Sigma_{c,c}}, \lambda_1, \lambda_2).
\end{align}
Also, $\mathcal{R}_{part, conv, rate} = \mathcal{R}_{ach, rate}$. Hence,
it follows that $V \leq U$.
\end{proof}

We can now prove Theorem \ref{thm : optimality}.

\begin{proof} \textit{of Theorem \ref{thm : optimality} : } Let $\mu \geq 1$. The proof of the theorem
follows directly from Lemmas \ref{lem : M = U}, \ref{lem : N = V} and \ref{lem : 6.5}. From Lemma \ref{lem : M = U}, we have that the optimum value of the optimization problem (\ref{eqn : optimization 1}), $M$ equals the optimum value of optimization problem (\ref{eqn : optimization 2}), $U$. From Lemma \ref{lem : N = V}, we have that the optimum value of optimization problem (\ref{eqn : optimization 6}), $N$ equals the optimum value of the optimization problem (\ref{eqn : optimization 7}), $V$. $M$ is the solution of the optimum $\mu R_p + R_c$ over the achievable region and $N$ is the solution of the optimum $\mu R_p + R_c$ over $\mathcal{R}_{part,out}^{\alpha}$ described in (\ref{eqn : partial converse region}). Hence if the condition given by (\ref{eqn : condition}) is satisfied for $\alpha^*$ given by (\ref{eqn : alpha^*}), $M \leq N$. From Lemma \ref{lem : 6.5}, we also have $V \leq U$. Hence, we have that the optimal value of the original optimization problem
(\ref{eqn : optimization 1}), $M$ is equal to the optimal value of
the optimization problem described by (\ref{eqn : optimization 6}),
$N$. Hence, the achievable region
$\mathcal{R}_{in}$ is $\mu$-sum optimal.
\end{proof}
\section{Numerical Results}\label{sec : numerical results}
In this section, we provide some numerical results on the capacity region of the MIMO cognitive channel. We consider a MIMO cognitive system where the licensed and cognitive transmitters have one antenna each, and the licensed and cognitive receivers have one and two antennas respectively. We assume that the channel coefficients are real and also restrict ourself to real inputs and outputs. We generate the channel values randomly
\begin{align}
\mathbf{H_{p,p}} = 1.4435, &\quad \mathbf{H_{p,c}} = \left[\begin{array}{c} -0.3510 \\ 0.6232 \end{array}\right],\nonumber\\
\mathbf{H_{c,p}} = 0.799, &\quad \mathbf{H_{c,c}} = \left[\begin{array}{c} 0.9409 \\ -0.9921 \end{array}\right].\nonumber
\end{align}
We assume a power constraint of $5$ at the licensed and cognitive transmitters. In Figure $8$, we plot the achievable region, $\mathcal{R}_{in}$ and the region $\mathcal{R}_{part, out}^{\alpha}$ for different values of $\alpha$.
\begin{figure}[hbtp]
\label{pl1}
\centering
\includegraphics[width=3 in]{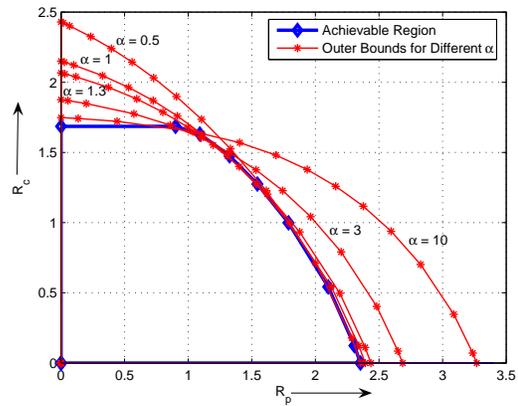}
\caption{Plot of Achievable Region $\mathcal{R}_{in}$ and partial outer bounds $\mathcal{R}_{part, out}^{\alpha}$ for different values of $\alpha$.}
\end{figure}
Figure $8$ shows how $\mathcal{R}_{part, out}^{\alpha}$ intersects with $\mathcal{R}_{in}$ at different points for different values of $\alpha$.

Next, we find the maximum value of rate than can be supported by the licensed user in the example we considered. In both the achievable region and the outer bound, this corresponds to maximizing the $\mu$-sum $\mu R_p + R_c$ when $\mu \rightarrow \infty$. This would correspond to using all the power to support the licensed user. Note that the maximum value of $R_p$ in the set described by $\mathcal{R}_{part,out}^{\alpha}$ is an upper bound on the maximum value of $R_p$ in the set $\mathcal{R}_{in}$ for all values of $\alpha > 0$, irrespective of the channel parameters.

Maximizing $R_p$ over $\mathcal{R}_{in}$ : The cognitive transmitter uses all its power for helping the licensed user. That is $\mathrm{Tr}(\mathbf{\Sigma_{c,p}}) = P_c$. This then reduces to a MIMO channel with channel matrix given by $\mathbf{G} = \left[\begin{array}{cc} \mathbf{H_{p,p}} & \mathbf{H_{c,p}}\end{array}\right]$. The licensed transmitter has a power constraint of $P_p$ and the cognitive transmitter has a power constraint of $P_c$. Applying this to our example channel, we have $\mathbf{G} = \left[\begin{array}{cc}1.4435 & 0.799\end{array}\right]$. The optimum covariance matrix is of the form
\begin{displaymath}
\mathbf{\Sigma_{p, net}} = \left[\begin{array}{cc} 5 & 5\rho \\ 5\rho & 5\end{array}\right],
\end{displaymath}
where $\rho$ is the correlation between the two transmitters. Therefore, the rate achieved by the licensed user is
\begin{displaymath}
R_p (\rho) = \frac{1}{2}\log(1 + \mathbf{G} \mathbf{\Sigma_{p, net}} \mathbf{G}^{\dagger}).
\end{displaymath}
The maximum rate is attained at $\rho = 1$ and the maximum value of $R_p$ is $2.3542$.

Maximizing $R_p$ over $\mathcal{R}_{part, out}^{\alpha}$ : For a given $\alpha$, this reduces to a single user MIMO channel with $\mathbf{G_{\alpha}} = \left[\begin{array}{cc}\mathbf{H_{p,p}} & \mathbf{H_{c,p}}/\sqrt{\alpha}\end{array}\right]$ and a sum power constraint of $P_p + \alpha P_c$. Note that, there is a significant difference between the two single user MIMO channels. The MIMO channel that we considered when solving the maximum value of $R_p$ in the achievable region had individual power constraints at the licensed and cognitive transmitters. However, the MIMO channel we obtain when solving for the maximum value of $R_p$ over $\mathcal{R}_{part, out}^{\alpha}$ has a sum power constraint. This is a conventional MIMO channel and the optimum covariance matrix is obtained by water-filling. For a given $\alpha$, the best $R_p$ is got by
\begin{eqnarray}
\begin{array}{c} \max R_p (\alpha) = \frac{1}{2}\log\left| \mathbf{I} + \mathbf{G_{\alpha}}\mathbf{\Sigma_{p, net}}\mathbf{G_{\alpha}}\right| \\
\textrm{such that } \mathrm{Tr}(\mathbf{\Sigma_{p, net}}) \leq P_p + \alpha P_c. \end{array}\nonumber
\end{eqnarray}
It is easy to solve this problem if we look at the flipped channel $\mathbf{G_{\alpha}^{\dagger}}$. The capacity of the flipped channel is given by
\begin{align}
R_p(\alpha) &= \frac{1}{2}\log\left| \mathbf{I} + \mathbf{G_{\alpha}^{\dagger}} (P_p + \alpha P_c) \mathbf{G_{\alpha}}\right|\nonumber\\
 &= \frac{1}{2}\log\left(1 + (P_p + \alpha P_c) \mathbf{G_{\alpha}}\mathbf{G_{\alpha}^{\dagger}}\right).\nonumber
\end{align}
Note that $R_p(\alpha)$ is an outer bound on the maximum value of $R_p$. The best upper bound is got by minimizing over all possible values of $\alpha$. The optimum value of $\alpha$ is got by solving a cubic equation $ 2(0.799)^2 \alpha^3 + (0.799)^2\alpha^2 - 1.4435^2 = 0$, and its approximate value is $0.9689$.
\section{Conclusions}\label{sec : conclusions}
In this paper, we derived an achievable region, $\mathcal{R}_{in}$ given by (\ref{eqn : achievable region})
and an outer bound, $\mathcal{R}_{out}^{\alpha, \mathbf{\Sigma_z}}$ given by (\ref{eqn : converse region}) for the MIMO
cognitive channel. We describe conditions when the achievable region is $\mu$-sum
optimal for any $\mu \geq 1$. In particular, for any $\mu \geq 1$, there exists $\alpha^* \in (0, \infty)$, such that if the region given by $\mathcal{R}_{part,out}^{\alpha^*}$ optimizes the $\mu-$ sum rate of the SMBC (for that particular $\alpha^*$), then the achievable region achieves the $\mu$-sum capacity of the MCC.

\end{document}